\documentclass[12pt]{article}%
\usepackage{subfigure}     
\usepackage{algorithm}
\usepackage{algpseudocode}
\usepackage{alphabeta}
\usepackage{amsmath}
\usepackage{amsfonts}
\usepackage{amssymb}
\usepackage{boxedminipage}
\usepackage[open,openlevel=1]{bookmark}
\usepackage{caption}
\usepackage{color}
\usepackage{graphicx}
\usepackage{listings}
\usepackage[most]{tcolorbox}
\usepackage{multirow}%
\usepackage{tikz}
\providecommand{\U}[1]{\protect \rule{.1in}{.1in}}
\newtheorem{theorem}{Theorem}[section]

\newtheorem{proposition}[theorem]{Proposition}

\newenvironment{proof}[1][Proof]{\noindent \textbf{#1.} }{\  \rule{0.5em}{0.5em}}
\begin{document}
\title{Static Nuel Games with Terminal Payoff}
\author{Symeon Mastrakoulis and Athanasios Kehagias}
\maketitle
\begin{abstract}
In \ this paper we study a variant of the \emph{Nuel} game (a generalization
of the \emph{duel}) which is played in turns by $N$ players. In each turn a
single player \emph{must} fire at one of the other players and has a certain
probability of hitting and killing his target. \ The players shoot in a fixed
sequence and when a player is eliminated, the \textquotedblleft
move\textquotedblright \ passes to the next\ surviving player. The winner is
the last surviving player. We prove that, for every $N\geq2$, the Nuel has a
stationary Nash equilibrium and provide algorithms for its computation.
\end{abstract}


\section{Introduction\label{sec01}}

In \ this paper we study a variant of the \emph{nuel} game (a generalization
of the \emph{duel}) which is played in turns by $N$ players. In each turn a
single player \emph{must} fire at one of the other players (in other words,
\emph{abstention} is not allowed)\ and has a certain probability of hitting
and killing his target. \ The players shoot in a fixed sequence and when a
player is eliminated, the \textquotedblleft move\textquotedblright \ passes to
the next\ surviving player. The winner is the last surviving player.

In what follows, we will use the term \textquotedblleft$N$%
-uel\textquotedblright \ to describe the $N$-player game; hence the $2$-uel or
\emph{duel} involves two players, the $3$-uel or \emph{truel} involves three
players etc.

Early works on the static $3$-uel are
\cite{Gardner1966,Kinnaird1946,Larsen1948,Mosteller1987,Shubik1954} in which
the postulated game rules guarantee the existence of \emph{exactly one}
survivor (\textquotedblleft winner\textquotedblright).{} A more general
analysis appears in \cite{Knuth1973} which considers the possibility of
\textquotedblleft cooperation\textquotedblright \ between the players. This
idea is further studied in
\cite{Kilgour1971,Kilgour1975,Kilgour1977,Zeephongsekul1991}. Recent papers on
the truel include
\cite{Amengual2005a,Amengual2005b,Amengual2006,Bossert2002,Brams1997,Brams2001,Brams2003,Dorraki2019,Toral2005,Wegener2021,Xu2012}%
. Only a few of these papers
\cite{Amengual2005a,Amengual2005b,Amengual2006,Zeephongsekul1991} make short
remarks on the general $N$-uel (i.e., for $N\geq3$) and of these only
\cite{Zeephongsekul1991} but the matter is not pursued in depth. \footnote{Let
us also note the existence of an extensive literature on a quite different
type of duel games, which essentially are \emph{games of timing}
\cite{Barron2013,Dresher1961,Karlin1959}. However, this literature is not
relevant to the game studied in this paper.}

By \emph{solving} the $N$-uel, we mean establishing the existence of one or
more \emph{Nash equilbria} (NE)\ and computing these equilbria. As will be
explained in a later section, solving the $N$-uel involves solving a system of
nonlinear equations or, equivalently, minimizing the total \emph{equilibrium
error}. Hence the problem is essentially one of global optimization.

The paper is organized as follows. In \ Section \ref{sec02} we present the
rules of the game. In Section \ref{sec03} we set up the $N$-uel \emph{system
of payoff equations} and prove that, for every $N$ and under appropriate
conditions, this system always has a unique solution, i.e., every $N$-uel has
uniquely defined payoffs for all players. In Section \ref{sec04} we prove that
for every $N$, under appropriate conditions, the $N$-uel has a 
\emph{stationary equilbrium strategy profile} and we provide algorithms to
compute this profile (equivalently, to solve the payoff equations system). In
Section \ref{sec05} we present illustrative experiments. In Section
\ref{sec06}\ we conclude and present some future research directions.

\section{The $N$-uel Game\label{sec02}}

We will now define the $N$-uel game rigorously for every $N\in \left \{
2,3,...\right \}  $. The game involves $N$ players, who will be denoted by
$P_{1},...,P_{N}$ or by $1,...,N$, and evolves in discrete times (turns)
$t\in \left \{  0,1,2,...\right \}  $. For each $n\in \{1,...,N\}$, $P_n$
has a \emph{marksmanship} $p_n$, which is the probability that he hits (and kills) 
the opponent whom he shoots. In what follows we assume that, for all  $n\in \{1,...,N\}$, 
$p_n$ is \emph{strictly} between 0 and 1.
In the next subsections we define the components of the game.

\subsection{States and Actions}

Each \emph{game state} has the form $\mathbf{s}=s_{0}s_{1}...s_{N}$, where
\ $s_{0}\in \left \{  0,1,...,N\right \}  $ and, for $n\in \left \{
1,...,N\right \}  $, $s_{n}=0$ (resp. 1) means that $P_{n}$ is dead (resp.
alive). Now suppose that at time $t\in \left \{  0,1,2,...\right \}  $ the game
state is $\mathbf{s}\left(  t\right)  =s_{0}\left(  t\right)  ...s_{N}\left(
t\right)  $. We have the following possibilities.

\begin{enumerate}
\item If $s_{0}\left(  t\right)  =n\in \left \{  1,...,N\right \}  $, then the
game is in progress and $P_{n}$ \textquotedblleft has the
move\textquotedblright, i.e., $P_{n}$ is the single player who will shoot in
the current turn. In this case we will also have $s_{1}\left(  t\right)
...s_{N}\left(  t\right)  \in \left \{  0,1\right \}  ^{N}$; 
$s_{i}\left( t\right)  =1$ means that $P_{i}$ is alive and $s_{i}\left(  t\right)  =0$
means that he is dead.

\item If $s_{0}\left(  t\right)  =0$, then the game has terminated (i.e.,
there is a single alive player)\ and we must actually have $s_{1}\left(
t\right)  ...s_{N}\left(  t\right)  = 0...010...0$; i.e., there exists a
single $n$ such that $s_{n}\left(  t\right)  =1$ ($P_{n}$ is the single alive
player)\ and, for $m\neq n$, $s_{m}\left(  t\right)  =0$ ($P_{m}$ is dead).
\end{enumerate}

\noindent We exclude from consideration \emph{inadmissible} states, i.e.,
states which will never be visited during a play of the game\footnote{For
example, with $N=3$ players, the state $s_{0}s_{1}s_{2}s_{3}=1011$ will never
occur because it corresponds to the dead player $P_{1}$ having the move.}. It
will be useful to define the following sets of \emph{admissible} states:%
\begin{align*}
\forall k  &  \in \left \{  2,...,N\right \}  :S_{k}=\left \{  \mathbf{s}
=s_{0}s_{1}...s_{N} :s_{0}\neq0\text{ and }s_{s_{0}}=1\text{ and }\sum
_{n=1}^{N}s_{n}=k\right \}  ,\\
\forall k  &  \in \left \{  1,...,N\right \}  :\widetilde{S}_{k}=\left \{
\mathbf{s}=0s_{1}...s_{N} :\text{and }s_{k}=1\text{ and }s_{n}=0\text{ for
}n\neq k\right \}  .
\end{align*}
$S_{k}$ is the set of states in which $k$ players are alive and $P_{s_{0}}$ is
the (alive)\ player who has the move; $\widetilde{S}_{k}$ is the set of states
in which the sole surviving player is $P_{k}$ and no player has the move.
Letting $S_{1}=%
{\displaystyle \bigcup \limits_{k=1}^{N}}
\widetilde{S}_{k}$, the set of all admissible states is $S=%
{\displaystyle \bigcup \limits_{k=1}^{N}}
S_{k}.$

\emph{Game actions} have the form $a\mathbf{=}m\in \left \{\lambda,1,...,N\right \}
$. Suppose that at time $t\in \left \{  1,2,...\right \}  $ the game action is
$a\left(  t\right)  =m$; if $m>0$ then the player who has the move (as
determined by the corresponding game state $\mathbf{s}$) will fire at $P_{m}$.
Depending on the current state, an admissible action must satisfy the
following:\ a player cannot fire either at himself or at an already dead
player. We write $a=\lambda$ to denote the \textquotedblleft \emph{null
action}\textquotedblright, i.e., no shooting takes place.

We assume that the game has \emph{perfect information}, i.e., at every $t$ all
players know all previous states and actions.

\subsection{State Transitions}

The game evolution is described in terms of state transitions.

\begin{enumerate}
\item The game starts at some initial state $\mathbf{s}\left(  0\right)
=s_{0}\left(  0\right)  s_{1}\left(  0\right)  s_{2}\left(  0\right)
...s_{N}\left(  0\right)  $.

\item Assume at $t\in \left \{  1,2,...\right \}  $ we have $\mathbf{s}\left(
t-1\right)  =s_{0}\left(  t-1\right)  s_{1}\left(  t-1\right)  ...,s_{N}%
\left(  t-1\right)  $ with $n=s_{0}\left(  t-1\right)  \neq0$ and $P_{n}$
performs action $a\left(  t\right)  =m\in \left \{  1,...,N\right \}  $ where
$s_{m}\left(  t-1\right)  =1$. Then the next state is $\mathbf{s}\left(
t\right)  =s_{0}\left(  t\right)  s_{1}\left(  t\right)  ...s_{N}\left(
t\right)  $ and is obtained by the following rules.

\begin{enumerate}
\item For the $s_{1}\left(  t\right)  ...s_{N}\left(  t\right)  $ part \ of
the state we have:%
\begin{align*}
&  \Pr \left(  s_{m}\left(  t\right)  =0\right)  =p_{n}\text{ (i.e., }%
P_{n}\text{ hit and killed }P_{m}\text{),}\\
&  \Pr \left(  s_{m}\left(  t\right)  =1\right)  =1-p_{n}\text{ (i.e., }%
P_{n}\text{ missed }P_{m}\text{),}\\
&  \forall i\in \left \{  1,...,m-1,m+1,...,N\right \}  :\Pr \left(  s_{i}\left(
t\right)  =s_{i}\left(  t-1\right)  \right)  =1,
\end{align*}
\emph{where we assume} \emph{that the marksmanships }$p_{n}\in \left(
0,1\right)  $ \emph{for all} $n\in \left \{  1,...,N\right \}  $.

\item The $s_{0}\left(  t\right)  $ part of the state specifies the next
player who has the move. This will be the \textquotedblleft next\ alive
player\textquotedblright \ and is best understood by some examples for the
three player case (analogous things hold for $N>3$ players).

\begin{enumerate}
\item Suppose $\mathbf{s}\left(  t-1\right)  =1111$, $a\left(  t\right)  =2$
and $s_{1}\left(  t\right)  s_{2}\left(  t\right)  s_{3}\left(  t\right)
=111$. This means that $P_{1}$ had the move, fired at $P_{2}$ and missed him;
hence the next player to have the move is $P_{2}$, i.e., $s_{0}\left(
t\right)  =2$.

\item Suppose $\mathbf{s}\left(  t-1\right)  =1111$, $a\left(  t\right)  =2$
and $s_{1}\left(  t\right)  s_{2}\left(  t\right)  s_{3}\left(  t\right)
=101$. This means that $P_{1}$ had the move, fired at $P_{2}$ and killed him;
hence the next player to have the move is $P_{3}$, i.e., $s_{0}\left(
t\right)  =3$.

\item Suppose $\mathbf{s}\left(  t-1\right)  =3111$, $a\left(  t\right)  =1$
and $s_{1}\left(  t\right)  s_{2}\left(  t\right)  s_{3}\left(  t\right)
=111$. This means that $P_{3}$ had the move, fired at $P_{1}$ and missed him;
hence the next player to have the move is $P_{1}$, i.e., $s_{0}\left(
t\right)  =1$.

\item Supposse $\mathbf{s}(t-1) = 1101$, $a(t) = 3$ and $s_1(t)s_2(t)s_3(t) = 100$. 
This means that $P_1$ had the move, fired at $P_3$ and killed him; hence there is no 
next player to have the move i.e., the game has terminated and $s_0(t) = 0$.
\end{enumerate}
\end{enumerate}

\item Finally, if $\mathbf{s}\left(  t-1\right)  =0...010...0$ then the game
has terminated and there is not either a next action $a\left(  t\right)  $ or
a next state\ $\mathbf{s}\left(  t\right)  $.
\end{enumerate}

Using the above rules, for every $N\in \left \{  2,3,..\right \}  $ we can
construct a \emph{controlled Markov chain} with state space $S$ and transition
probability matrix $\Pi \left(  a\right)  $, with%
\begin{align*}
\forall \mathbf{s}^{\prime}  &  \in S\backslash S_{1}:\Pi_{\mathbf{s}^{\prime
},\mathbf{s}^{\prime \prime}}\left(  a\right)  =\Pr \left(  \mathbf{s}\left(
t\right)  =\mathbf{s}^{\prime \prime}|\mathbf{s}\left(  t-1\right)
=\mathbf{s}^{\prime},a\left(  t\right)  =a\right)  ,\\
\forall \mathbf{s}^{\prime}  &  \in S_{1}:\Pi_{\mathbf{s}^{\prime}%
,\mathbf{s}^{\prime}}\left(  \lambda \right)  =1.
\end{align*}
The state sequence obtained from the above controlled Markov chain corresponds
exactly to the state sequence of a $N$-uel, except for the fact that, in the
Markov chain, every terminal state $\mathbf{s}^{\prime}\in S_{1}$ loops back
to itself with probability one, producing an infinite state sequence, while in
the $N$-uel as soon as a terminal state is entered the game is over, resulting
in a finite state sequence. Note that when $\mathbf{s}^{\prime}\in S_{1}$ the
only admissible action is $a=\lambda$ (no-shoot action). \noindent

Because in each turn a player must shoot (unless he is the sole survivor)\ and
$p_{n}>0$ for all $n\in \left \{  1,...,N\right \}  $, it is easy to prove the following.

\begin{proposition}
\label{prop0201}\normalfont For all $N$, the probability that the $N$-uel
terminates in finite time equals one.
\end{proposition}

The above complete the description of the $N$-uel, which we will also denote
by $\Gamma_{N}\left(  \mathbf{p}\right)  $, where $\mathbf{p}=\left(
p_{1},...,p_{N}\right)  $. 

\subsection{Histories, Payoffs and Strategies}

We now proceed to define and discuss
\emph{histories}, \emph{payoffs} and \emph{strategies} used in $\Gamma_{N}\left(  \mathbf{p}\right)  $.
In view of Proposition \ref{prop0201}, we only need to consider finite length
histories. The game history at time $t$ is $\mathbf{h}\left(  t\right)
\mathbf{=}\mathbf{s}\left(  0\right)  a\left(  1\right)  \mathbf{s}\left(
1\right)  ...\mathbf{s}\left(  t\right)  $ and its \emph{length} is $t$. The
set of all \emph{admissible terminal histories} is%
\[
H_{0}=\left \{  \mathbf{h}\left(  t\right)  =\mathbf{s}\left(  0\right)
a\left(  1\right)  \mathbf{s}\left(  1\right)  ...\mathbf{s}\left(  t\right)
:\text{ }\mathbf{h}\left(  t\right)  \text{ can occur in the game
and }\mathbf{s}\left(  t\right)  \in S_{1}\right \}  .
\]
The set of all \emph{admissible nonterminal histories} is
\[
H_{1}=\left \{  \mathbf{h}\left(  t\right)  =\mathbf{s}\left(  0\right)
a\left(  1\right)  \mathbf{s}\left(  1\right)  ...\mathbf{s}\left(  t\right)
:\mathbf{h}\left(  t\right)  \text{ can occur in the game
}\mathbf{s}\left(  t\right)  \not \in S_{1}\right \}  .
\]
The set of all \emph{admissible histories} is $H=H_{0}\cup H_{1}$.

For $n\in \left \{  1,...,N\right \}  $, $P_{n}$'s \emph{payoff} function is
$\mathsf{Q}_{n}:H\rightarrow \mathbb{R}$ and is defined as follows:
\[
\forall n\in \left \{  1,...,N\right \}  ,\forall \mathbf{h}\in H:\mathsf{Q}%
_{n}\left(  \mathbf{h}\right)  =\left \{
\begin{array}
[c]{ll}%
1 & \text{iff }\mathbf{h=}\mathbf{s}\left(  0\right)  a\left(  1\right)
  ...\mathbf{s}\left(  t\right)  \in H_{0}\text{ and
}s_{n}\left(  t\right)  =1\text{,}\\
0 & \text{otherwise.}%
\end{array}
\right.
\]
I.e., upon game termination, the single surviving player receives one payoff unit.

Let $\Delta_{N}$ denote the set of $N$-long probability vectors $\mathbf{x}%
=\left(  x_{1},...,x_{N}\right)  $. A \emph{strategy} is a function
$\sigma:H_{1}\rightarrow \Delta_{N}\cup \left \{  \left(  0,...,0\right)
\right \}  $. Suppose $\mathbf{h} = \mathbf{s}(0)a(1)...a(t-1)ns1...sN$ and $P_{n}$ uses strategy 
$\mathbf{x}_{n}=\left(  x_{n1},...,x_{nN}\right)  =\sigma_{n}\left(  \mathbf{h}\right)$, then
\[
\forall m\in \left \{  1,...,N\right \}  :x_{nm}=\Pr \left(
\text{\textquotedblleft}P_{n}\text{ shoots }P_{m}\text{\textquotedblright%
}\right)  .
\]
Note that a strategy needs to be defined only for nonterminal histories, since
at the end of a terminal history no shooting takes place. We will only
consider \emph{admissible strategies}, i.e., those which assign positive
probability only to admissible actions\footnote{For example a strategy which
assigns positive probability to shooting a dead player is inadmissible.}. In
addition, when $P_{n}$ uses an admissible $\sigma_{n}$ and $m\neq n$, we must
have
\[
\forall \mathbf{h}=\mathbf{s}\left(  0\right)  a\left(  1\right)  ...a\left(
t-1\right)  ms_{1}...s_{N}:\sigma_{n}\left(  \mathbf{h}\right)  =\left(
0,...,0\right)
\]
(since $P_{n}$ does not have the move he cannot fire at anybody).

The probabilities $\mathbf{x=}\left(  \mathbf{x}_{1},...,\mathbf{x}%
_{N}\right)  =\sigma_{n}\left(  \mathbf{h}\right)  $ depend, in general, on
the entire game history (up to the current time). A \emph{stationary strategy}
is one that only depends on the current state and then we write $\sigma
_{n}\left(  \mathbf{s}\right)  =\mathbf{x}_{n}\mathbf{=}\left(  x_{n1}%
,...,x_{nN}\right)  $. A \emph{pure strategy} is a function that maps
histories to probability vectors which concentrate all probability to a single
component, equal to one; with a slight abuse of notation we then also write
$\sigma_{n}\left(  \mathbf{h}\right)  =m$, meaning that $P_{n}$ shoots $P_{m}$
with probability one ($\sigma_{n}\left(  \mathbf{h}\right)  =\left(
0,...,0\right)  $ can be abbreviated as $\sigma_{n}\left(  \mathbf{h}\right)
=\lambda$, meaning that $P_{n}$ does not shoot). A \emph{pure stationary
strategy} is a pure strategy which only depends on the current state.

A \emph{strategy profile} is a vector $\sigma=\left(  \sigma_{1}%
,...,\sigma_{N}\right)  $, where $\sigma_{n}$ is the strategy used by $P_{n}$
(for $n\in \left \{  1,...,N\right \}  $). An initial state $\mathbf{s}\left(
0\right)  $ and a strategy profile $\sigma$, define a probability measure on
all histories, hence the \emph{expected payoff} to $P_{n}$ is well defined by
\[
Q_{n}\left(  \mathbf{s}\left(  0\right)  ,\sigma \right)  =\mathbb{E}%
_{\mathbf{s}\left(  0\right)  ,\sigma}\left(  \mathsf{Q}_{n}\left(
\mathbf{h}\right)  \right)  .
\]
Note that $Q_{n}\left(  \mathbf{s}\left(  0\right)  ,\sigma \right)  $ equals
the probability that $P_{n}$ is the sole remaining survivor or, equivalently,
the winner of $\Gamma_{N}\left(  \mathbf{p}\right)  $.

\section{$N$-uel Payoff System\label{sec03}}

For each $N\in \left \{  2,3,...\right \}  $ and each fixed strategy profile
$\sigma$, the payoffs $\left(  Q_{i}\left(  \mathbf{s},\sigma \right)  \right)
_{i\in\{1,...,N\},\mathbf{s}\in S}$ satisfy a system of \emph{payoff equations}. In this
section we will formulate the \emph{payoff system}, study its form and
properties, and will introduce appropriate algorithms for its solution. For
clarity of presentation, we will first study the $2$-uel and $3$-uel as
special cases and then deal with the general case of the $N$-uel.

\subsection{The $2$-uel}

With $N=2$ players we have
\begin{align*}
S_{1}  &  =\left \{  010,001\right \}  \text{ and }S_{2}=\left \{
111,211\right \}  ,\\
S  &  =S_{1}\cup S_{2}=\left \{  010,001,111,211\right \}  .
\end{align*}
For each player, the only admissible strategy is to keep shooting at his
opponent, until one player is eliminated. The only nonzero state transition
probabilities are%
\begin{align*}
\Pr \left(  \mathbf{s}\left(  t\right)  =010|\mathbf{s}\left(  t-1\right)
=111,a\left(  t\right)  =2\right)   &  =p_{1}\\
\Pr \left(  \mathbf{s}\left(  t\right)  =211|\mathbf{s}\left(  t-1\right)
=111,a\left(  t\right)  =2\right)   &  =1-p_{1}\\
\Pr \left(  \mathbf{s}\left(  t\right)  =001|\mathbf{s}\left(  t-1\right)
=211,a\left(  t\right)  =1\right)   &  =p_{2}\\
\Pr \left(  \mathbf{s}\left(  t\right)  =111|\mathbf{s}\left(  t-1\right)
=211,a\left(  t\right)  =1\right)   &  =1-p_{2}%
\end{align*}
We can represent the above information compactly by the following state
transition graph, where actions and corresponding transition probabilities are
written next to the respective edges.

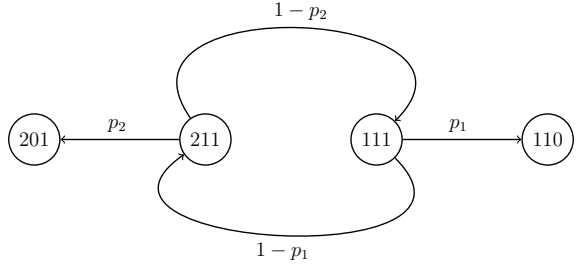
\begin{figure}[ptbh]
\centering \scalebox{0.65} {
\begin{tikzpicture}[node distance={35mm}, thick, main/.style = {draw, circle}]
\node[main] (1) {$111$};
\node[main] (2) [left of=1]{$211$};
\node[main] (3) [right of=1] {$110$};
\node[main] (4) [left of=2]{$201$};
\draw[->](1) to [out=315,in=215,looseness=2,below] node{$1-p_1$}(2);
\draw[->](2) to [out=125,in=45,looseness=2,above] node{$1-p_2$}(1);
\draw[->] (1) -- node[midway, above right, sloped, pos=0.33] {$p_1$} (3);
\draw[->] (2) -- node[midway, above right, sloped, pos=0.66] {$p_2$} (4);
\end{tikzpicture}
}\caption{The $2$-uel state transition graph.}%
\end{figure}

\noindent The partition of states into the sets $S_{1}$ and $S_{2}$ is
reflected in the structure of the state transition graph:\ each $\mathbf{s}\in
S_{2}$ has two successors:\ a state $\mathbf{s}^{\prime}\in S_{2}$ and a state
$\mathbf{s}^{\prime \prime}\in S_{1}$; each $\mathbf{s}\in S_{1}$ has no
outgoing transitions.

Since the players have no choice of strategies, we do not have have here a
strategic game, but a \emph{game of chance}. In fact, the $2$-uel is
essentially a Markov chain\footnote{Except for the fact that the terminal
states have no associated state transitions.}. For every initial state, a
player's payoff equals his winning probability, which we will now compute. For
notational brevity, we set
\[
\forall i,n\in \left \{  1,2\right \},
\text{ for all admissible } ns_{1},s_{2}:
V_{i,ns_{1}s_{2}}=Q_{i}\left(  ns_{1}s_{2},\left(  \sigma
_{1}^{\prime},\sigma_{2}^{\prime}\right)  \right)  ,
\]
i.e., the payoff to $P_{i}$ when the game starts in state $\mathbf{s}\left(
0\right)  \mathbf{=}ns_{1}s_{2}$, and both players use the same strategy
$\sigma_{1}^{\prime}=\sigma_{2}^{\prime}$ of always shooting at the opponent.
It is easily seen that the $V_{1,ns_{1}s_{2}}$'s satisfy the following system
of equations:
\begin{align}
V_{1,010}  &  =1\nonumber \\
V_{1,001}  &  =0\nonumber \\
V_{1,111}  &  =\left(  1-p_{1}\right)  V_{1,211}+p_{1}V_{1,010}\label{eq03005}%
\\
V_{1,211}  &  =\left(  1-p_{2}\right)  V_{1,111}+p_{2}V_{1,200}\nonumber
\end{align}
This system has the unique solution
\begin{equation}
V_{1,010}=1,\quad V_{1,001}=0,\quad V_{1,111}=\frac{p_{1}}{p_{1}+p_{2}%
-p_{1}p_{2}},\quad V_{1,211}=\frac{p_{1}\left(  1-p_{2}\right)  }{p_{1}%
+p_{2}-p_{1}p_{2}}\allowbreak. \label{eq03001}%
\end{equation}
A similar system can be set up for the $V_{2,ns_{1}s_{2}}$ variables and has
the unique solution
\begin{equation}
V_{2,010}=0,\quad V_{2,001}=1,\quad V_{2,111}=\frac{p_{2}\left(
1-p_{1}\right)  }{p_{1}+p_{2}-p_{1}p_{2}},\quad V_{2,211}=\frac{p_{2}}%
{p_{1}+p_{2}-p_{1}p_{2}}\allowbreak. \label{eq03002}%
\end{equation}
The formulas (\ref{eq03001})-(\ref{eq03002})\ provide the solution to the
$\Gamma_{2}\left(  \mathbf{p}\right)  $, i.e., the winning probability for
each player and for each starting state.

\subsection{The $3$-uel}

We will limit our analysis to the case where all players use stationary
strategies. Suppose $\sigma_{n}\left(  \mathbf{s}\right)  $ is a stationary
strategy used by $P_{n}$. This can be characterized as follows.

\begin{enumerate}
\item For all $\mathbf{s}\in S_{1}$ (one player alive)\ there is no need to
define $\sigma_{n}\left(  \mathbf{s}\right)  $.

\item For all $\mathbf{s}\in S_{2}$ (two players alive)\ the only admissible
strategy $\sigma_{n}\left(  \mathbf{s}\right)  $ is to shoot at the sole alive opponent.

\item Finally, consider $\sigma_{n}\left(  \mathbf{s}\right)  =\mathbf{x}%
_{n}=\left(  x_{n1},x_{n2},x_{n3}\right)  $ when $\mathbf{s=}s_{0}s_{1}%
s_{2}s_{3}\in S_{3}$.

\begin{enumerate}
\item When $s_{0}\neq n$, we will necessarily have $\mathbf{x}_{n}=\left(
0,0,0\right)  $.

\item When $s_{0}=n$, we will necessarily have $x_{nn}=0$.
\end{enumerate}
\end{enumerate}

\noindent Hence, for every $n\in \left \{  1,2,3\right \}  $, an admissible
stationary strategy $\sigma_{n}$ for $P_{n}$ is fully determined by the two
positive numbers%
\[
\forall m\neq n:x_{nm}=\Pr \left(  \text{\textquotedblleft}P_{n}\text{ fires at
}P_{m}\text{\textquotedblright}|\text{\textquotedblleft the game state is
}n111\text{\textquotedblright}\right)
\]
which must satisfy $\sum_{m\neq n}x_{nm}=1$. Using the $x_{mn}$'s we can draw
the state transition graph shown in Figure 2.

\begin{minipage}[c]{0.95\linewidth}
\centering
\scalebox{0.65} {
\begin{tikzpicture}[node distance={35mm}, thick, main/.style = {draw, circle}]
\node[main] (1) {$1111$};
\node[main] (2) [below left of=1]{$2111$};
\node[main] (3) [below right of=1] {$3111$};
\node[main] (4) [above left of=1]{$2110$};
\node[main] (5) [above right of=1] {$3101$};
\node[main] (6) [above left of=2]{$1110$};
\node[main] (7) [below of=2] {$3011$};
\node[main] (8) [above right of=3]{$1101$};
\node[main] (9) [below of=3] {$2011$};
\node[main](10)[above of=1]{$1100$};
\node[main](11)[below left of=2]{$2010$};
\node[main](12)[below right of=3]{$3001$};
\draw[->] (1) -- node[midway, above right, sloped, pos=0.66] {\small{$1-p_1$}} (2);
\draw[->] (2) -- node[midway, above right, sloped, pos=0.33] {\small{$1-p_2$}} (3);
\draw[->] (3) -- node[midway, above right, sloped, pos=0.66] {\small{$1-p_3$}} (1);
\draw[->] (1) -- node[midway, above right, sloped, pos=0.66] {\small{$x_{13}p_1$}} (4);
\draw[->] (2) -- node[midway, above right, sloped, pos=0.66] {\small{$x_{23}p_2$}} (6);
\draw[->] (3) -- node[midway, above right, sloped, pos=0.33] {\small{$x_{32}p_3$}} (8);
\draw[->] (1) -- node[midway, above right, sloped, pos=0.33] {\small{$x_{12}p_1$}} (5);
\draw[->] (2) -- node[midway, above right, sloped, pos=0.33] {\small{$x_{21}p_2$}} (7);
\draw[->] (3) -- node[midway, above right, sloped, pos=0.33] {\small{$x_{31}p_3$}} (9);
\draw [->] (4) to [out=180,in=90,looseness=1] node{\small{$1-p_2$}} (6);
\draw [->] (7) to [out=315,in=225,looseness=1] node{\small{$1-p_1$}}  (9);
\draw [->] (8) to [out=90,in=0,looseness=1]  node{\small{$1-p_1$}} (5);
\draw [->] (6) to [out=0,in=270,looseness=1] node{\small{$1-p_3$}}  (4);
\draw [->] (9) to [out=135,in=45,looseness=1]  node{\small{$1-p_2$}} (7);
\draw [->] (5) to [out=270,in=180,looseness=1]  node{\small{$1-p_3$}} (8);
\draw[->](4) to [out=135,in=135,looseness=2] node{\small{$p_2$}}(11);
\draw[->](9) to [out=270,in=225,looseness=2] node{\small{$p_2$}}(11);
\draw[->](8) to [out=0,in=45,looseness=2] node{\small{$p_1$}}(10);
\draw[->](6) to [out=180,in=135,looseness=2] node{\small{$p_1$}}(10);
\draw[->](5) to [out=45,in=45,looseness=2] node{\small{$p_3$}}(12);
\draw[->](7) to [out=270,in=315,looseness=2] node{\small{$p_3$}}(12);
\end{tikzpicture}
}

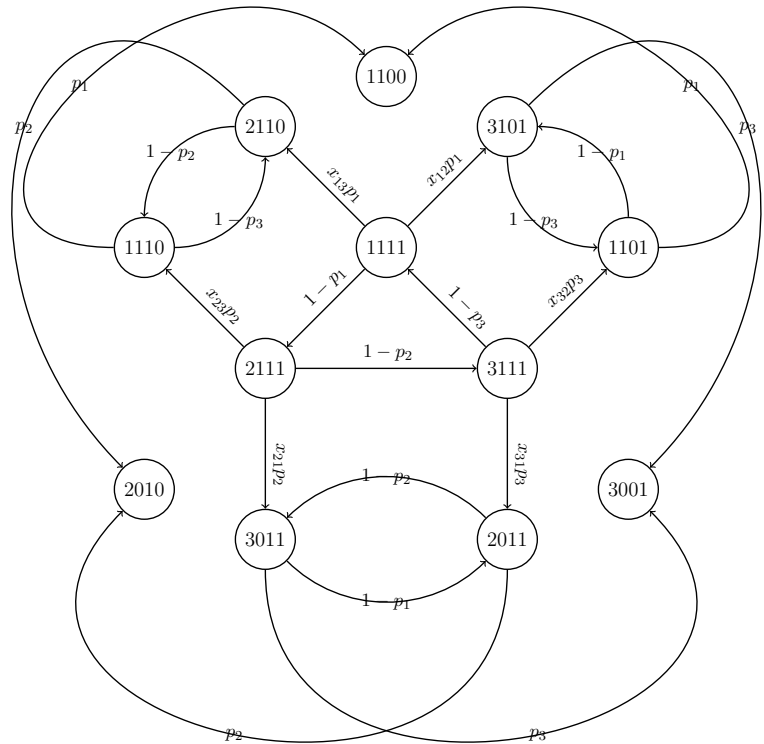
\captionof{figure}{The $3$-uel state transition graph.}%
\end{minipage}

\bigskip

\noindent We will now show the following.

\begin{proposition}
\normalfont For any $p_{1}$, $p_{2}$, $p_{3}$ such that $m\neq n\Rightarrow
p_{m}\neq p_{n}$, and every admissible strategy profile $\sigma$, the $3$-uel
payoff system has a unique solution.
\end{proposition}

\begin{proof}
Our goal is to compute each player's expected payoff (equivalently, his
winning probability). To this end, similarly to the two-player game, we define%
\[
\forall i,n\in \left \{  1,2,3\right \},
\text{for all admissible } ns_{1}s_{2}s_{3}:
V_{i,ns_{1}s_{2}s_{3}}=Q_{i}\left(  ns_{1}s_{2}s_{3},\sigma \right),
\]
i.e., the payoff to $P_{i}$ when the game starts at state $ns_{1}s_{2}s_{3}$
and the strategy profile $\sigma=\left(  \sigma_{1},\sigma_{2},\sigma
_{3}\right)  $ is used\footnote{The dependence on $\sigma$ is omitted from the
notation, for the sake of brevity.}.

Let us momentarily focus on $P_{1}$'s payoffs. It is easily seen that the
variables $\left(  V_{1,ns_{1}s_{2}s_{3}}\right)  _{ns_{1}s_{2}s_{3}\in S}$
must satisfy the following equations:%
\begin{align}
V_{1,0100}  &  =1\nonumber \\
V_{1,0010}  &  =0\nonumber \\
V_{1,0001}  &  =0\nonumber \\
V_{1,1101}  &  =(1-p_{1})V_{1,3101}+p_{1}V_{1,0100}\nonumber \\
V_{1,3101}  &  =(1-p_{3})V_{1,1101}+p_{3}V_{1,0001}\nonumber \\
V_{1,1110}  &  =(1-p_{1})V_{1,2110}+p_{1}V_{1,0100}\nonumber \\
V_{1,2110}  &  =(1-p_{2})V_{1,1110}+p_{2}V_{1,0010}\label{eq03004}\\
V_{1,2011}  &  =(1-p_{2})V_{1,3011}+p_{2}V_{1,0010}\nonumber \\
V_{1,3011}  &  =(1-p_{3})V_{1,2011}+p_{3}V_{1,0001}\nonumber \\
V_{1,1111}  &  =(1-p_{1})V_{1,2111}+x_{12}p_{1}V_{1,3101}+x_{13}%
p_{1}V_{1,2110}\nonumber \\
V_{1,2111}  &  =(1-p_{2})V_{1,3111}+x_{23}p_{2}V_{1,1110}+x_{21}%
p_{2}V_{1,3011}\nonumber \\
V_{1,3111}  &  =(1-p_{3})V_{1,1111}+x_{31}p_{3}V_{1,2011}+x_{32}%
p_{3}V_{1,1101}\nonumber
\end{align}
The above system can be solved in a stepwise fashion. The first three
equations immediately yield the values of $V_{1,0100}$, $V_{1,0100}$,
$V_{1,0100}$. The fourth and fifth equations can be solved to obtain the
values of $V_{1,1101}$ and $V_{1,3101}$:%
\begin{equation}
V_{1,1101}=\frac{p_{1}}{p_{1}+p_{3}-p_{1}p_{3}},\quad V_{1,3101}=\frac
{p_{1}\left(  1-p_{3}\right)  }{p_{1}+p_{3}-p_{1}p_{3}}; \label{eq03003}%
\end{equation}
naturally, these are exactly the payoffs for a duel between $P_{1}$ and
$P_{3}$. Similarly, the sixth and seventh equations yield $V_{1,1110}$ and
$V_{1,2110}$ (expressions for $V_{1,s_{0}s_{1}s_{2}s_{3}}$ similar to those of
(\ref{eq03003})) and the eight and ninth equations yield $V_{1,2011}%
=V_{1,3011}=0$.

The final three equations involve the unknowns $V_{1,1111}$, $V_{1,2111}$,
$V_{1,3111}$ and the previously computed $V_{1,ns_{1}s_{2}s_{3}}$'s. The
system can be rewritten as
\begin{align*}
V_{1,1111}-(1-p_{1})V_{1,2111}  &  =x_{12}p_{1}V_{1,3101}+x_{13}%
p_{1}V_{1,2110},\\
V_{1,2111}-(1-p_{2})V_{1,3111}  &  =x_{23}p_{2}V_{1,1110}+x_{21}%
p_{2}V_{1,3011},\\
V_{1,3111}-(1-p_{3})V_{1,1111}  &  =x_{31}p_{3}V_{1,2011}+x_{32}%
p_{3}V_{1,1101}.
\end{align*}
Letting
\begin{align}
A_{1}  &  =x_{12}p_{1}V_{1,3101}+x_{13}p_{1}V_{1,2110},\nonumber \\
A_{2}  &  =x_{23}p_{2}V_{1,1110}+x_{21}p_{2}V_{1,3011},\label{eq03041}\\
A_{3}  &  =x_{31}p_{3}V_{1,2011}+x_{32}p_{3}V_{1,1101},\nonumber
\end{align}
the system becomes
\[
\left[
\begin{array}
[c]{ccc}%
1 & -\left(  1-p_{1}\right)  & 0\\
0 & 1 & -\left(  1-p_{2}\right) \\
-\left(  1-p_{3}\right)  & 0 & 1
\end{array}
\right]  \left[
\begin{array}
[c]{c}%
V_{1,1111}\\
V_{2,1111}\\
V_{3,1111}%
\end{array}
\right]  =\left[
\begin{array}
[c]{c}%
A_{1}\\
A_{2}\\
A_{3}%
\end{array}
\right]
\]
The determinant of the system is $D=1-\left(  1-p_{1}\right)  \left(
1-p_{2}\right)  \left(  1-p_{3}\right)  >0$. Hence the system has a unique
solution which is
\begin{align}
V_{1,1111}  &  =\frac{A_{1}+A_{2}+A_{3}-p_{1}A_{2}-p_{2}A_{3}-p_{1}A_{3}%
+p_{1}p_{2}A_{3}}{1-\left(  1-p_{1}\right)  \left(  1-p_{2}\right)  \left(
1-p_{3}\right)  }\label{eq03042}\\
V_{1,2111}  &  =\frac{A_{2}+A_{3}+A_{1}-p_{2}A_{3}-p_{3}A_{1}-p_{2}A_{1}%
+p_{2}p_{3}A_{1}}{1-\left(  1-p_{1}\right)  \left(  1-p_{2}\right)  \left(
1-p_{3}\right)  }\label{eq03043}\\
V_{1,3111}  &  =\frac{A_{3}+A_{1}+A_{2}-p_{3}A_{1}-p_{1}A_{2}-p_{3}A_{2}%
+p_{3}p_{1}A_{2}}{1-\left(  1-p_{1}\right)  \left(  1-p_{2}\right)  \left(
1-p_{3}\right)  } \label{eq03044}%
\end{align}
In the same manner we can prove that the payoff systems for $P_{2}$ and
$P_{3}$ have unique solutions and this completes the proof. Note that $\left(
V_{i,n111}\right)  _{i,n\in \left \{  1,2,3\right \}  }$ are actually functions
of $\mathbf{x}_{1}$, $\mathbf{x}_{2}$, $\mathbf{x}_{3}$; for brevity, the
dependence has been suppressed from the notation.
\end{proof}

The structure and stepwise solution of the payoff equations correspond to the
structure of the state transition diagram. Namely, the vertices of the state
transition graph are the game states and the possible transitions are as
follows:\ $S_{1}$ states are terminal, each $S_{2}$ state can transit either
to a single other $S_{2}$ state or to a single $S_{1}$ state, and each $S_{3}$
state can transit either to a single other $S_{3}$ state or one of two $S_{2}$
states. Furthermore, the $S_{3}$ states form a cycle, i.e.,
\[
...\rightarrow \left(  1111\right)  \rightarrow \left(  2111\right)
\rightarrow \left(  3111\right)  \rightarrow \left(  1111\right)  \rightarrow
...
\]
These facts, clearly, correspond to the stepwise procedure of solving the
payoff system. For each $n\in \left \{  1,2,3\right \}  $ separately, we first
obtain the payoffs $\left(  V_{n,\mathbf{s}}\right)  _{\mathbf{s\in}S_{1}}$,
then the $\left(  V_{n,\mathbf{s}}\right)  _{\mathbf{s\in}S_{2}}$ and finally
the $\left(  V_{n,\mathbf{s}}\right)  _{\mathbf{s\in}S_{3}}$.

\subsection{The $N$-uel}

We will now write and solve the payoff system which, analogously to
(\ref{eq03005}) and (\ref{eq03004}),\ governs \ the $N$-uel expected total
payoffs. We use, for all $n\in \left \{  1,...,N\right \}  $ and all
$ms_{1}...s_{N}\in S$, the notation
\[
V_{n,ms_{1}...s_{N}}=Q_{n}\left(  ms_{1}...s_{N},\sigma \right)  .
\]
In addition we introduce the following notations. For all $\mathbf{s=}%
ms_{1}...s_{N}\in S$, we define
\begin{align*}
\text{the set of alive players}  &  :L\left(  \mathbf{s}\right)  =\left \{
n:s_{n}=1\right \}  ,\\
\text{the set of alive players other than }m  &  :L_{m}\left(  \mathbf{s}%
\right)  =\left \{  n:s_{n}=1\text{ and }n\neq m\right \}  .
\end{align*}
For all $\mathbf{s}=ms_{1}...s_{N}\in S$, we will use
$\mathbf{N}\left(  \mathbf{s}\right)  $ to denote the state following $\mathbf{s}$
when no player is killed. 
For example, when $N=4$ we have \
\[
\mathbf{N}\left(11111\right)=21111,\quad
\mathbf{N}\left(11011\right)=31011,\quad 
\mathbf{N}\left(41111\right)=11111,\quad... \quad .
\]
Furthermore, for all $\mathbf{s}=ms_{1}...s_{N}\in S$, we will use
$\mathbf{N}_i\left(  \mathbf{s}\right)  $ to denote the 
the state following $\mathbf{s}$
when $P_i$ is killed. 
For example, when $N=4$ we have \
\[
\mathbf{N}_2\left(11111\right)=31011,\quad 
\mathbf{N}_4\left(31111\right)=11110,\quad... \quad .
\]
Finally, for all $n\in \left \{  1,...,N\right \}  $ and all $\mathbf{s}%
=ms_{1}...s_{N}\in S$, the probability that $P_{m}$ shoots $P_{n}$, when the
state is $\mathbf{s}$ and $P_{m}$ uses $\sigma_{m}$, is:
\[
x_{ms_{1}...s_{N},n}=\Pr \left(  a=n|P_{m}\text{ uses strategy }\sigma
_{m}\text{, current state is }\mathbf{s}=ms_{1}...s_{N}\right)  .
\]

\noindent Using the above notations and assuming a given strategy profile
$\sigma$ (which determines all the shooting probabilities $x_{ms_{1}...s_{N},n}$) 
the expected total payoffs for $P_{1}$ satisfy the following equations.

\begin{enumerate}
\item At terminal states we have:%
\begin{equation}
V_{1,0s_{1}...s_{N}}=1 \text{ when } s_1=1 \text{ and } V_{1,0s_{1}...s_{N}}=0 \text{ when } s_1=0. \label{eq03006}%
\end{equation}

\item At all admissible states with two alive players $P_1$ and $P_m$ we have:%
\begin{equation}
V_{1,1s_{1}...s_{N}}=\frac{p_{1}}{p_{1}+p_{m}-p_{1}p_{m}} 
\text{ and }
V_{1,ms_{1}...s_{N}}=\frac{p_{1}\left(1-p_{m}\right)  }{p_{1}+p_{m}-p_{1}p_{m}},
\label{eq03007}%
\end{equation}
which are $P_{1}$'s winning probabilities in a duel against $P_{m}$.
\item At all admissible states $is_{1}...s_{N}$ with two alive players, both different from $P_1$, we have:%
\begin{equation}
V_{1,is_{1}...s_{N}}=0
\label{eq03007a}
\end{equation}

\item At all admissible states with more than two alive players: 
for all $k\in \left \{3,...,N\right \}  $, $\mathbf{s}=ms_{1}...s_{N}\in S_{k}$, we have:
\begin{align}
\text{when }1  &  \in L\left(  \mathbf{s}\right)  :V_{1,ms_{1}...s_{N}}=
\left(1-p_{m}\right) V_{1,\mathbf{N}\left(\mathbf{s}\right)}+
\sum_{n\in L_{m}\left(  \mathbf{s}\right)  }
x_{ms_{1}...s_{N},n}p_{m}V_{1,\mathbf{N}_n\left(\mathbf{s}\right)},
\label{eq03008}\\
\text{when }1  &  \not \in L\left(  \mathbf{s}\right)  :V_{1,ms_{1}...s_{N}%
}=0. \label{eq03009}%
\end{align}

\end{enumerate}

\noindent The payoff system which must be satisfied by the $V_{1,ms_{1}%
...s_{N}}$'s consists of the equations (\ref{eq03006})-(\ref{eq03009}).
Similar systems are satisfied by the $V_{n,ms_{1}...s_{N}}$'s, for
$n\in \left \{  2,...,N\right \}  $.

\begin{proposition}
\normalfont For every $N\in \left \{  2,3,...\right \}  $, for any $p_{1}$, ...,
$p_{N}$ such that $m\neq n\Rightarrow p_{m}\neq p_{n}$, and every admissible
strategy profile $\sigma$, the $N$-uel payoff system has a unique solution.
\end{proposition}

\begin{proof}
We only consider the payoff systems regarding $\left(  V_{1,\mathbf{s}%
}\right)  _{\emph{s}\in S}$ (the cases $\left(  V_{n,\mathbf{s}}\right)
_{\emph{s}\in S}$ with $n\geq2$ are treated similarly). The proof is by
induction. Clearly, for $N=2$, the payoff system has a unique solution, given
by (\ref{eq03001}). Now suppose the $\left(  N-1\right)  $-uel payoff system
has a unique solution and consider the $N$-uel payoff system (\ref{eq03006}%
)-(\ref{eq03009}).

Take any $K\in \left \{  2,...,N-1\right \}  $; for any state $\mathbf{s=}%
s_{0}s_{1}...s_{N}\in S_{K}$, we want to determine (for all $n\in \left \{
1,...,N\right \}  $)\ the corresponding $V_{1,s_{0}s_{1}...s_{N}}=Q_{1}\left(
s_{0}s_{1}...s_{N},\sigma \right)  $. This is $P_{1}$'s payoff in a $N$-uel
involving himself, $K-1$ other alive and $N-K$ dead players, which is the same
as $P_{1}$'s payoff in a $K$-uel involving himself and $K-1$ other alive
players. Hence $V_{n,s_{0}s_{1}...s_{N}}$ can be computed by solving the
respective $K$-uel with $K$ alive players and relabeling the $K$-uel players
and their payoffs so as to correspond with the $K$ alive players of the
$N$-uel. By the inductive assumption, the $K$-uel has a unique solution, hence
the value $V_{1,s_{0}s_{1}...s_{N}}$ is also uniquely determined.

Consequently the $V_{1,s_{0}s_{1}...s_{N}}$'s are uniquely determined for all
$\mathbf{s}=s_{0}s_{1}...s_{N}\in \cup_{K=2}^{N-1}S_{K}$. It remains to show
that the $V_{1,s_{0}s_{1}...s_{N}}$'s with $\mathbf{s}=s_{0}s_{1}...s_{N}\in
S_{N}$ are also uniquely determined. Now, $\mathbf{s}\in S_{N}$ means there
exists $N$ alive players; hence $s_{1}=...=s_{N}=1$ and there exist \ exactly
$N$ such states:%
\[
S_{N}=\left \{  11...1,...,N1...1\right \}  .
\]
Also, since in such states all players are alive, we have
\[
21...1=\mathbf{N}\left(11...1\right),\quad
31...1=\mathbf{N}\left(21...1\right),\quad
...,\quad
11...1=\mathbf{N}\left(N1...1\right).
\]
Each of the above states appears once on the left side of an equation
(\ref{eq03008}) and once on the right side of another equation (\ref{eq03008}%
). Let us rename the corresponding $V_{1,s_{0}s_{1}...s_{N}}$ variables as
follows.
\[
\forall m\in \left \{  1,...,N\right \}  :Z_{m}=V_{1,ms_{1}...s_{N}}.
\]
Let us also define%
\[
\forall m\in \left \{  1,...,N\right \}  :
A_{m}=\sum_{i\in L_{m}\left(\mathbf{s}\right)  }x_{ms_{1}...s_{N},i}p_{m}
V_{1,\mathbf{N}_i\left(ms_{1}...s_{N}\right)}.
\]
It follows that, for all $\mathbf{s}\in S_{N}$, the equations (\ref{eq03008})
can be rewritten in the form
\begin{equation}
\left[
\begin{array}
[c]{ccccc}%
1 & -\left(  1-p_{1}\right)  & 0 & ... & 0\\
0 & 1 & -\left(  1-p_{2}\right)  & ... & 0\\
0 & 0 & 1 & ... & 0\\
... & ... & ... & ... & ...\\
-\left(  1-p_{N}\right)  & 0 & 0 & ... & 1
\end{array}
\right]  \left[
\begin{array}
[c]{c}%
Z_{1}\\
Z_{2}\\
Z_{3}\\
...\\
Z_{N}%
\end{array}
\right]  =\left[
\begin{array}
[c]{c}%
A_{1}\\
A_{2}\\
A_{3}\\
...\\
A_{N}%
\end{array}
\right]  \label{eq03011}%
\end{equation}
Furthermore, since the states $\mathbf{N}_i\left(ms_{1}...s_{N}\right)\in S_{N-1}$, the
$V_{1,\mathbf{N}_i\left(ms_1...s_N\right)}$'s are uniquely determined as solutions of an
$\left(  N-1\right)  $-uel.

A necessary and sufficient condition for the system (\ref{eq03011}) to have a
unique solution is that the determinant%
\[
D_{N}=\left \vert
\begin{array}
[c]{ccccc}%
1 & -\left(  1-p_{1}\right)  & 0 & ... & 0\\
0 & 1 & -\left(  1-p_{2}\right)  & ... & 0\\
0 & 0 & 1 & ... & 0\\
... & ... & ... & ... & ...\\
-\left(  1-p_{N}\right)  & 0 & 0 & ... & 1
\end{array}
\right \vert
\]
is different from zero. It is easily proved that
\[
D_{N}=1-\left(  1-p_{1}\right)  \left(  1-p_{2}\right)  ...\left(
1-p_{N}\right)  .
\]
Since, by assumption, for all $n$ we have $p_{n}\in \left(  0,1\right)  $, it
follows that $D_{N}>0$ and the inductive step is completed.
\end{proof}

\bigskip

\noindent Now our problem is to solve the system of payoff equations and the
above proof suggests a solution method. In what follows we define, for every
$i\in \left \{  1,...,N\right \}  $ and every $S^{\prime}\subseteq S$, the
\emph{payoffs vector} $\mathbf{V}_{i,S^{\prime}}=\left(  V_{i,\mathbf{s}%
}\right)  _{\mathbf{s}\in S^{\prime}}$, i.e., the vector of all payoffs
indexed by players and states. For example, let $N=3$; then the sets of all
states with one and two surviving players are, respectively,
\[
S_{1}=\left \{  0100,0010,0001\right \}, \quad S_{2}=\left \{  1110,2110,1101,3101,2011,3011\right \}
\]
and the corresponding payoff vectors are 
\begin{align*}
\mathbf{V}_{1,S_{1}}  &  =\left(  V_{1,0100},V_{2,0100},V_{3,0100}%
,V_{1,0010},V_{2,0010},V_{3,0010},V_{1,0001},V_{2,0001},V_{3,0001}\right) \\
\mathbf{V}_{1,S_{2}}  &  =\left(  V_{1,1110},V_{2,1110},V_{3,1110}%
,V_{1,2110},...,V_{2,3011},V_{3,3011}\right)  .
\end{align*}
Initialization of the $\mathbf{V}_{i,S_{1}}$'s, i.e., the set of payoffs (to
all players)\ for states with a single surviving player, is immediate; for
example when $N=3$ we have
\begin{align*}
V_{1,0100}  &  =1,\quad V_{2,0100}=0,\quad V_{3,0100}=0,\\
V_{1,0010}  &  =0,\quad V_{2,0010}=1,\quad V_{3,0010}=0,\\
V_{1,0001}  &  =0,\quad V_{2,0001}=0,\quad V_{3,0001}=1,
\end{align*}
and similar values are obtained for $\mathbf{V}_{2,S_{1}}$ and $\mathbf{V}%
_{3,S_{1}}$.

The function \textsc{SolveNuel}, presented below in pseudocode, computes the
$\mathbf{V}_{i,S}$ vectors (for $i\in \left \{  1,...,N\right \}  $)\ as follows. \ 

\begin{enumerate}
\item The function inputs are:\ 

\begin{enumerate}
\item the number of players $N$,

\item the player of interest $i$,

\item the vector of marksmanships $\mathbf{p}=\left(  p_{1},...,p_{N}\right)
$ and

\item the strategy profile $\sigma$.
\end{enumerate}

\item At inialization, the one-player payoff vector $\mathbf{V}_{i,S_{1}}$ is computed.

\item In the outer loop of the function, $K$ is the number of living players,
from $K=2$ to $K=N$; for each $K$ we create the set $\mathcal{C}$ of
$\binom{N}{K}$ combinations of living players.

\item In the inner loop, we solve a $K$-uel for each player set $C=\left \{
n_{1},n_{2},...,n_{K}\right \}  \in \mathcal{C}$. This involves solving a system
of the $K$ unknown $\mathbf{V}_{S_{C}}$'s; when obtained these are used to
gradually populate the elements of the \textquotedblleft
full\textquotedblright \ $\mathbf{V}_{i,S}$ vector.

\item On completion of both loops, all components of the $\mathbf{V}_{i,S}$
vector have been computed and the function returns $\mathbf{V}_{i,S}$.
\end{enumerate}

\bigskip

\begin{minipage}[c]{0.95\linewidth}
\begin{algorithm}[H]
\caption{Function for recursive $N$-uel solution}
\label{alg:cap}
\begin{algorithmic}
\Function{SolveNuel}{$N,i,\mathbf{p},\sigma$}
\State Construct the state set $S$ corresponding to the $N$-uel.
\State Compute $\mathbf{V}_{i,S_{1}}$
\For{$K=2..N$}
\State Let $\mathcal{C}$ be the set of all combinations of $K$ players out of $\left \{  1,...,N\right \}  $
\For{$C\in \mathcal{C}$}
\State Let $S_{C}$ be the set of states corresponding to player subset $C$
\State Compute $\mathbf{V}_{i,S_{C}}$ by solving a $K$-uel
\EndFor
\EndFor
\State Return $\mathbf{V}_{i,S}$
\EndFunction
\end{algorithmic}
\end{algorithm}
\end{minipage}\bigskip

The $\mathbf{V}_{i,S_{k}}$ values, for $k\in \left \{  3,...,N\right \}  $, are
obtained by solving $\binom{N}{K}$ systems, each involving $K$ unknowns. Exact
values can be obtained by Cramer's rule or matrix inversion. However, we have
found that implementation is simpler when the following iterative algorithm is
used. We first present the algorithm and then prove its correctness.

\begin{enumerate}
\item The function inputs are:\ 

\begin{enumerate}
\item the number of alive players $K$,

\item the player of interest $i$,

\item the vector of marksmanships $\mathbf{p}=\left(  p_{1},...,p_{N}\right)
$,

\item the strategy profile $\sigma$,

\item the payoffs vector for $K-1$ players $\mathbf{V}_{i,S_{K-1}}$ and

\item the termination parameter $\varepsilon$.
\end{enumerate}

\item We initialize, for all states $ms_{1}...s_{N}\in S_{K}$, at arbitrary
values $V_{i,ms_{1}...s_{N}}^{\left(  0\right)  }$.

\item Then we iterate, for $t\in \left \{  0,1,2,...\right \}  $ and for each
state $ms_{1}...s_{N}\in S_{K}$, to obtain new $V_{i,ms_{1}...s_{N}}^{\left(
t+1\right)  }$ values by (\ref{eq03012}).

\item If at some iteration $t$ we have $\max_{i\in \left \{  1,...,K\right \}
,\mathbf{s}\in S_{K}}\left \vert V_{i,\mathbf{s}}^{\left(  t+1\right)
}-V_{i,\mathbf{s}}^{\left(  t\right)  }\right \vert <\varepsilon$, the
algorithm terminates and returns $\mathbf{V}_{i,S}=\mathbf{V}_{i,S}^{\left(
t+1\right)  }$.
\end{enumerate}

\bigskip

\begin{minipage}[c]{0.95\linewidth}
\begin{algorithm}[H]
\caption{Iterative Solution of Payoff System}
\label{alg:cap}
\begin{algorithmic}
\Function{IterSolve}{$K,i,\mathbf{p},\sigma,\mathbf{V}_{S_{K-1}},\varepsilon$}
\For{$\mathbf{s}=ms_{1}...s_{N}\in S_{K}$}
\State $V_{i,ms_{1}...s_{N}}^{\left(0\right)}$ arbitrary
\EndFor
\For{$t\in \left \{0,1,2,...\right \}  $}
\For{$\mathbf{s}=ms_{1}...s_{N}\in S_{K}$}
\State
\begin{equation}
V_{i,\mathbf{s}}^{\left(t+1\right)}
=
\left(1-p_{m}\right)V^{(t)}_{i,\mathbf{N}(\mathbf{s})}+
\sum_{n\in L_m(\mathbf{s})}x_{\mathbf{s},n}p_{m}V_{i,\mathbf{N}_n(\mathbf{s})}
\label{eq03012}
\end{equation}
\EndFor
\If{$\max_{\mathbf{s}\in S_{K}}\left \vert V_{i,\mathbf{s}}^{\left(  t+1\right)}-V_{i,\mathbf{s}}^{\left(  t\right)  }\right \vert <\varepsilon$}
\State Break
\EndIf
\EndFor
\State
$\mathbf{V}_{i,S_K}=\mathbf{V}_{i,S_K}^{\left(t+1\right)  }$
\State Return $\mathbf{V}_{i,S_K}$
\EndFunction
\end{algorithmic}
\end{algorithm}
\end{minipage}\bigskip

\noindent The following proposition shows that, for given strategy profile
$\sigma$, the IPC\ algorithm yields in the limit the payoffs of the $N$-uel.

\noindent

\begin{proposition}
\normalfont For every $K\in \left \{  2,3,...\right \}  $, $i\in \left \{
1,2,...,K\right \}  $, for any $p_{1}$, ..., $p_{K}$ such that $m\neq
n\Rightarrow p_{m}\neq p_{n}$, and for every admissible strategy profile
$\sigma$, the iterative solution of the payoff system always converges and we
have%
\[
\forall \mathbf{s\in S}_{K}:\lim_{t\rightarrow \infty}V_{i,\mathbf{s}}^{\left(
t+1\right)  }=V_{i,\mathbf{s}}.
\]

\end{proposition}

\begin{proof}
Consider the same iteration starting from two different initial conditions
$V_{i,S_{K}}^{\left(  0\right)  }$ and $U_{i,S_{K}}^{\left(  0\right)  }$.
Then we have%
\begin{align*}
\forall \mathbf{s}  &  =ms_{1}...s_{N}\in S_{K}:
V_{i,\mathbf{s}}^{\left(t+1\right)  }=\left(  1-p_{m}\right)  V_{i,\mathbf{N}(\mathbf{s})}^{(t)}+
\sum_{n\in L_{m}(\mathbf{s})}x_{\mathbf{s},n}p_{m}V_{i,\mathbf{N}_n(\mathbf{s})}\\
\forall \mathbf{s}  &  =ms_{1}...s_{N}\in S_{K}:
U_{i,\mathbf{s}}^{\left(t+1\right)  }=\left(  1-p_{m}\right)  U_{i,\mathbf{N}(\mathbf{s})}^{(t)}+
\sum_{n\in L_{m}(\mathbf{s})}x_{\mathbf{s},n}p_{m}V_{i,\mathbf{N}_n(\mathbf{s})}
\end{align*}
We then have
\begin{align*}
&  \forall \mathbf{s}=ms_{1}...s_{N}\in S_{K}:\left \vert V_{i,\mathbf{s}%
}^{\left(  t+1\right)  }-U_{i,\mathbf{s}}^{\left(  t+1\right)  }\right \vert=
\left(1-p_{m}\right)  \left \vert V_{i,\mathbf{N}(\mathbf{s})}^{(t)}-U_{i,\mathbf{N}(\mathbf{s})}^{(t)}\right \vert
\Rightarrow \\
&  \sum_{\mathbf{s}\in S_{K}}\left \vert V_{i,\mathbf{s}}^{\left(  t+1\right)
}-U_{i,\mathbf{s}}^{\left(  t+1\right)  }\right \vert \leq \sum_{\mathbf{s}\in
S_{K}}\left(  1-\min_{m}p_{m}\right)  \left \vert V_{i,\mathbf{N}%
(\mathbf{s})}^{(t)}-U_{i,\mathbf{N}(\mathbf{s})}^{(t)}\right \vert \Rightarrow \\
&  \sum_{m\in L\left(  \mathbf{s}\right)  }\left \vert V_{i,\mathbf{s}%
}^{\left(  t+1\right)  }-U_{i,\mathbf{s}}^{\left(  t+1\right)  }\right \vert
\leq \left(  1-\min_{m}p_{m}\right)  \sum_{\mathbf{s}\in S_{K}}\left \vert
V_{i,\mathbf{N}(\mathbf{s})}^{(t)}-U_{i,\mathbf{N}(\mathbf{s})}^{(t)}\right \vert \Rightarrow \\
&  \sum_{m\in L\left(  \mathbf{s}\right)  }\left \vert V_{i,\mathbf{s}%
}^{\left(  t+1\right)  }-U_{i,\mathbf{s}}^{\left(  t+1\right)  }\right \vert
\leq \left(  1-\min_{m}p_{m}\right)  ^{t+1}\sum_{\mathbf{s}\in S_{K}}\left \vert
V_{i,\mathbf{N}(\mathbf{s})}^{(0)}-U_{i,\mathbf{N}(\mathbf{s})}^{(0)}\right \vert .
\end{align*}
Since $\left \vert 1-\min_{m}p_{m}\right \vert \in \left(  0,1\right)  $, we
have
\[
\lim_{t\rightarrow \infty}\sum_{m\in L\left(  \mathbf{s}\right)  }\left \vert
V_{i,\mathbf{s}}^{\left(  t+1\right)  }-U_{i,\mathbf{s}}^{\left(  t+1\right)
}\right \vert =0
\]
This means that, for every $i$ and $\mathbf{s}$, the iteration tends to a
unique limit%
\[
\forall i,\mathbf{s}:\overline{V}_{i,\mathbf{s}}=\lim_{t\rightarrow \infty}%
\sum_{m\in L\left(  \mathbf{s}\right)  }V_{i,\mathbf{s}}^{\left(  t\right)  }%
\]
and we have
\[
\forall i,\mathbf{s}:\overline{V}_{i,\mathbf{s}}=\left(  1-p_{m}\right)
\overline{V}_{i,\mathbf{N}(\mathbf{s})}+
\sum_{n\in L_{m}(\mathbf{s})}x_{\mathbf{s},n}p_{m}V_{i,\mathbf{N}_n(\mathbf{s})}.
\]
In other words $\overline{V}_{i,S_{K}}$ satisfies the payoff equations, which
means that the iteration yields the unique solution of the payoff system.
\end{proof}

It should be pointed out that Algorithms 1 and 2 are \emph{guaranteed} to work
(i.e., solve the $N$-uel)\ when the $p_{n}$ marksmanships belong to $\left(
0,1\right)  $; but the algorithms \emph{may} also work even when some of the
$p_{n}$'s are equal to zero or to one.

\section{$N$-uel Stationary Equilibria\label{sec04}}

We are now ready to study the existence of $N$-uel equilibria. For every $N$
we have to solve a separate system of \emph{nonlinear} equations. For clarity
of presentation, we will first deal with the $3$-uel and then for the general
$N$-uel.

\begin{proposition}
\normalfont For any $p_{1}$, $p_{2}$, $p_{3}$ such that $m\neq n\Rightarrow
p_{m}\neq p_{n}$, the $3$-uel has a unique stationary deterministic Nash equilibrium
$\widehat{\sigma}$, which can be described as follows
\[
\widehat{\sigma}\left(  1111\right)  =\left(  0,\widehat{x}_{12},1-\widehat
{x}_{12}\right)  ,\quad \widehat{\sigma}\left(  2111\right)  =\left(
1-\widehat{x}_{23},0,\widehat{x}_{23}\right)  ,\quad \widehat{\sigma}\left(
3111\right)  =\left(  \widehat{x}_{31},1-\widehat{x}_{31},0\right)  ,
\]
where%
\[
\widehat{x}_{12}=\left \{
\begin{array}
[c]{ll}%
1 & \text{iff }p_{2}>p_{3}\\
0 & \text{else}%
\end{array}
\right.  ,\quad \widehat{x}_{23}=\left \{
\begin{array}
[c]{ll}%
1 & \text{iff }p_{3}>p_{1}\\
0 & \text{else}%
\end{array}
\right.  ,\quad \widehat{x}_{31}=\left \{
\begin{array}
[c]{ll}%
1 & \text{iff }p_{1}>p_{2}\\
0 & \text{else}%
\end{array}
\right.  .
\]
In other words, when in equilibrium, at every turn each player shoots at his
\textquotedblleft strongest\textquotedblright \ opponent with probability one.
\end{proposition}

\begin{proof}
It suffices to consider $P_{1}$'s point of view. His only strategy choice is
when the game is in state $\mathbf{s}=1111$ (in all other states, there exists
a unique admissible strategy), i.e., $P_{1}$ must choose $x_{12}$ and $x_{13}$
(subject to $x_{12}\geq0$, $x_{13}\geq0$, $x_{12}+x_{13}=1$) so as to maximize
$V_{1,1111}$, $V_{1,2111}$ and $V_{1,3111}$. Recall that $V_{1,1111}$ is given
from (\ref{eq03042}):
\[
V_{1,1111}=\frac{A_{1}+A_{2}+A_{3}-p_{1}A_{2}-p_{2}A_{3}-p_{1}A_{3}+p_{1}%
p_{2}A_{3}}{1-\left(  1-p_{1}\right)  \left(  1-p_{2}\right)  \left(
1-p_{3}\right)  }%
\]
where $A_{1},A_{2},A_{3}$ are given from (\ref{eq03041}). Note that $x_{12}$
and $x_{13}$ only appear in $A_{1}$. Hence $P_{1}$ wants to maximize
\begin{align*}
F_{1}\left(  x_{12},x_{13}\right)   &  =\frac{A_{1}}{1-\left(  1-p_{1}\right)
\left(  1-p_{2}\right)  \left(  1-p_{3}\right)  }\\
&  =\frac{x_{12}p_{1}V_{1,3101}+x_{13}p_{1}V_{1,2110}}{1-\left(
1-p_{1}\right)  \left(  1-p_{2}\right)  \left(  1-p_{3}\right)  }\\
&  =\frac{x_{12}p_{1}\frac{p_{1}\left(  1-p_{3}\right)  }{p_{1}+p_{3}%
-p_{1}p_{3}}+x_{13}p_{1}\frac{p_{1}\left(  1-p_{2}\right)  }{p_{1}+p_{2}%
-p_{1}p_{2}}}{1-\left(  1-p_{1}\right)  \left(  1-p_{2}\right)  \left(
1-p_{3}\right)  }\\
&  =p_{1}^{2}\frac{x_{12}\left(  1-p_{3}\right)  \left(  p_{1}+p_{2}%
-p_{1}p_{2}\right)  +x_{13}\left(  1-p_{2}\right)  \left(  p_{1}+p_{3}%
-p_{1}p_{3}\right)  }{\left(  1-\left(  1-p_{1}\right)  \left(  1-p_{2}%
\right)  \left(  1-p_{3}\right)  \right)  \left(  p_{1}+p_{2}-p_{1}%
p_{2}\right)  \left(  p_{1}+p_{3}-p_{1}p_{3}\right)  }%
\end{align*}
subject to the constraints $x_{12}+x_{13}=1,x_{12}\geq0,x_{13}\geq0$. Since
the denominator is positive, it suffices to choose $x_{12}$ (and consequently
$x_{13}=1-x_{12}$) so as to maximize%
\[
x_{12}\left(  1-p_{3}\right)  \left(  p_{1}+p_{2}-p_{1}p_{2}\right)
+x_{13}\left(  1-p_{2}\right)  \left(  p_{1}+p_{3}-p_{1}p_{3}\right)  .
\]
Finally, since
\[
\left(  1-p_{3}\right)  \left(  p_{1}+p_{2}-p_{1}p_{2}\right)  -\left(
1-p_{2}\right)  \left(  p_{1}+p_{3}-p_{1}p_{3}\right)  =\allowbreak
p_{2}-p_{3}%
\]
we have
\begin{align*}
p_{2}  &  >p_{3}\Leftrightarrow \left(  1-p_{3}\right)  \left(  p_{1}%
+p_{2}-p_{1}p_{2}\right)  >\left(  1-p_{2}\right)  \left(  p_{1}+p_{3}%
-p_{1}p_{3}\right) \\
p_{2}  &  <p_{3}\Leftrightarrow \left(  1-p_{3}\right)  \left(  p_{1}%
+p_{2}-p_{1}p_{2}\right)  <\left(  1-p_{2}\right)  \left(  p_{1}+p_{3}%
-p_{1}p_{3}\right)
\end{align*}
Hence the optimization rule for $V_{1,1111}$ is simple:%
\begin{align}
\text{if }p_{2}  &  >p_{3}\text{ then }\widehat{x}_{12}=1\text{, }\widehat
{x}_{13}=0\label{eq03032}\\
\text{if }p_{2}  &  <p_{3}\text{ then }\widehat{x}_{12}=0\text{, }\widehat
{x}_{13}=1\nonumber
\end{align}
This rule also maximizes
\begin{align*}
V_{1,2111}  &  =\frac{A_{2}+A_{3}+A_{1}-p_{2}A_{3}-p_{3}A_{1}-p_{2}A_{1}%
+p_{2}p_{3}A_{1}}{1-\left(  1-p_{1}\right)  \left(  1-p_{2}\right)  \left(
1-p_{3}\right)  }\\
&  =\frac{A_{1}-p_{3}A_{1}-p_{2}A_{1}+p_{2}p_{3}A_{1}+A_{2}+A_{3}-p_{2}A_{3}%
}{1-\left(  1-p_{1}\right)  \left(  1-p_{2}\right)  \left(  1-p_{3}\right)
}\\
&  =\frac{\left(  1-p_{2}\right)  \left(  1-p_{3}\right)  A_{1}+A_{2}%
+A_{3}-p_{2}A_{3}}{1-\left(  1-p_{1}\right)  \left(  1-p_{2}\right)  \left(
1-p_{3}\right)  }%
\end{align*}
and%
\begin{align*}
V_{1,3111}  &  =\frac{A_{3}+A_{1}+A_{2}-p_{3}A_{1}-p_{1}A_{2}-p_{3}A_{2}%
+p_{3}p_{1}A_{2}}{1-\left(  1-p_{1}\right)  \left(  1-p_{2}\right)  \left(
1-p_{3}\right)  }\\
&  =\frac{A_{1}-p_{3}A_{1}+A_{3}+A_{2}-p_{1}A_{2}-p_{3}A_{2}+p_{3}p_{1}A_{2}%
}{1-\left(  1-p_{1}\right)  \left(  1-p_{2}\right)  \left(  1-p_{3}\right)
}\\
&  =\frac{\left(  1-p_{3}\right)  A_{1}+A_{3}+A_{2}-p_{1}A_{2}-p_{3}%
A_{2}+p_{3}p_{1}A_{2}}{1-\left(  1-p_{1}\right)  \left(  1-p_{2}\right)
\left(  1-p_{3}\right)  }%
\end{align*}
Hence the rule (\ref{eq03032})\ simulataneously maximizes $V_{1,1111}$,
$V_{1,2111}$ and $V_{1,3111}$. This is the \textquotedblleft strongest
opponent\textquotedblright \ strategy and is $P_{1}$'s best response to any
$P_{2}$ and $P_{3}$ strategy. By a similar analysis we can prove that the
\textquotedblleft strongest opponent\textquotedblright \ strategy is also
$P_{2}$'s and $P_{3}$'s best response. This completes the proof.
\end{proof}

\begin{proposition}
\normalfont For every $N\in \left \{  3,4,...,\right \}$, the $N$-uel has a
stationary deterministic Nash equilibrium.
\end{proposition}

\begin{proof}
It suffices to consider $P_{1}$'s problem of determining his equilibrium
strategy $\widehat{\sigma}_{1}$. The proof will be inductive.

For all states $\mathbf{s}\in S_{3}$, $P_{1}$ can determine the value of his
optimal strategy $\widehat{\sigma}_{1}\left(  \mathbf{s}\right)  $ as follows.
For every such state in which he is not alive he has no strategy choice. For
every state in which he is alive, he must solve a $3$-uel against the other
two alive players; he can do this without considering the value of
$\widehat{\sigma}_{1}\left(  \mathbf{s}\right)  $ for states $\mathbf{s}%
\not \in S_{1}\cup S_{2}\cup S_{3}$.

Suppose that $P_{1}$ has determined $\widehat{\sigma}_{1}\left(
\mathbf{s}\right)  $ for all $\mathbf{s}\in \cup_{k=3}^{K-1}S_{k}$; now $\ $he
wants to determine $\  \widehat{\sigma}_{1}\left(  \mathbf{s}\right)  $ for all
states $\mathbf{s}\in S_{K}$. He has nothing to determine for states
$\mathbf{s}\in S_{K}$ in which he is not alive. There exist $\binom{N-1}{K-1}$
sets of states with $P_{1}$ and $K-1$ other players are alive. Let $S^{\prime
}$ be any such set and let the alive players be $n_{1},...,n_{K}$; in
particular, let $n_{1}=1$. With an appropriate state reordering, we can write
$S^{\prime}$ as
\[
S^{\prime}=\left \{  \mathbf{s}_{1},...,\mathbf{s}_{K}\right \}
\]
where $P_{n_{k}}$ is the player having the move in $\mathbf{s}_{k}$; in
particular, $\mathbf{s}_{1}$ is the state in which $P_{1}$ has the move. Now,
letting
\[
\forall k\in \left \{  1,...,K\right \}  :\left \{
\begin{array}
[c]{l}%
Z_{k}=V_{1,\mathbf{s}_{k}}\\
A_{k}=\sum_{i\in L_{k}\left(  \mathbf{s}_{k}\right)  }x_{\mathbf{s}_{k},i}p_{n_{k}}
V_{1,\mathbf{N}_i\left(\mathbf{s}_{k}\right)}
\end{array}
\right.  ,
\]
the following payoff system must be satisfied%
\begin{equation}
\left[
\begin{array}
[c]{cccc}%
1 & -\left(  1-p_{n_{1}}\right)  & ... & 0\\
0 & 1 & ... & 0\\
... & ... & ... & ...\\
-\left(  1-p_{n_{K}}\right)  & 0 & ... & 1
\end{array}
\right]  \left[
\begin{array}
[c]{c}%
Z_{1}\\
Z_{2}\\
...\\
Z_{K}%
\end{array}
\right]  =\left[
\begin{array}
[c]{c}%
A_{1}\\
A_{2}\\
...\\
A_{K}%
\end{array}
\right]  . \label{eq03021}%
\end{equation}
The $Z_{k}$'s are the unknowns and, since each state $\mathbf{N}_i\left(\mathbf{s}_{k}\right)$ 
belongs to $S_{K-1}$, the $A_{k}$'s are known (but depending on
the $x_{\mathbf{s}_{k},i}$'s). $P_{1}$ wants to maximize $Z_{k}%
=V_{1,\mathbf{s}_{k}}$ (for all $k\in \left \{  1,...,K\right \}  $). We can
solve (\ref{eq03021})\ using Cramer's rule. We have
\[
\left \vert
\begin{array}
[c]{cccc}%
1 & 1-p_{n_{1}} & ... & 0\\
0 & 1 & ... & 0\\
... & ... & ... & ...\\
1-p_{n_{K}} & 0 & ... & 1
\end{array}
\right \vert =1-\left(  1-p_{n_{1}}\right)  \left(  1-p_{n_{2}}\right)
...\left(  1-p_{n_{K}}\right)  >0,
\]
and, expanding with respect to the first column, we have
\begin{equation}
V_{1,\mathbf{s}_{1}}=Z_{1}=\frac{\left \vert
\begin{array}
[c]{ccccc}%
A_{1} & 1-p_{n_{1}} & 0 & ... & 0\\
A_{2} & 1 & 1-p_{n_{2}} & ... & 0\\
A_{3} & 0 & 1 & ... & 0\\
... & ... & ... & ... & ...\\
A_{K} & 0 & 0 & ... & 1
\end{array}
\right \vert }{1-\left(  1-p_{n_{1}}\right)  \left(  1-p_{n_{2}}\right)
...\left(  1-p_{n_{K}}\right)  } \label{eq03034}%
\end{equation}
or
\begin{align*}
V_{1,\mathbf{s}_{1}}=Z_{1}  &  =\frac{A_{1}\left \vert
\begin{array}
[c]{cccc}%
1 & 1-p_{n_{2}} & ... & 0\\
0 & 1 & ... & 0\\
... & ... & ... & ...\\
0 & 0 & ... & 1
\end{array}
\right \vert -A_{2}D_{2}+A_{3}D_{3}-...}{1-\left(  1-p_{n_{1}}\right)  \left(
1-p_{n_{2}}\right)  ...\left(  1-p_{n_{K}}\right)  }\\
&  =\frac{
\sum_{i\in L_{1}\left(\mathbf{s}_{1}\right)}x_{\mathbf{s}_{1},i}p_{n_{1}}
V_{1,\mathbf{N}_i\left(\mathbf{s}_{1}\right)}+B
}
{1-\left(  1-p_{n_{1}}\right)  \left(  1-p_{n_{2}}\right)  ...\left(  1-p_{n_{K}}\right)  }.
\end{align*}
In the above, $D_{k}$ is the determinant of the submatrix obtained by removing
the first column and the $k$-th row ($k\in \left \{  2,...,K\right \}  $) in
(\ref{eq03034})\ and $B=\sum_{k=2}^{K}\left(  -1\right)  ^{k-1}A_{k}D_{k}$;
neither the $A_{k}$'s nor $B$ contain $P_{1}$'s shooting probabilities
$x_{\mathbf{s}_{1},i}$. Hence $P_{1}$ needs only to maximize 
$\sum_{i\in L_{1}\left(\mathbf{s}_1\right)}
x_{\mathbf{s}_{1},i}p_{n_{1}}
V_{1,\mathbf{N}_i\left(\mathbf{s}_{1}\right)}$. 
The rule to achieve this is simple:
\begin{align}
\widehat{i}  &  =\arg \max_{i\in L_{1}\left(\mathbf{s}_1\right)}
V_{1,\mathbf{N}_i\left(\mathbf{s}_{1}\right)},\label{eq03033}\\
\forall i  &  \neq \widehat{i}:x_{\mathbf{s}_{1},i}=0\text{ and }%
x_{\mathbf{s}_{1},\widehat{i}}=1.\label{eq03034}
\end{align}
In (\ref{eq03033}), $\arg \max_{i}$ is understood as the smallest $i$ which 
achieves the maximum (there may exist more than one such and this may result in more than one Nash equilibria).
After some additional algebra
it can be verified that this rule also maximizes $V_{1,\mathbf{s}_{k}}$ for
all remaining $k\in \left \{  2,...,K\right \}  $.

This completes the inductive proof for $P_{1}$'s equilibrium strategy
$\sigma_{1}$. The proof for all other players works the same way. 
Let us define a \emph{family} of rules $\mathbf{R}_k$ (for $K\in \left \{  2,...,N\right \}  $):
\begin{quote}
\hspace*{-5mm}$\mathbf{R}_{K}$: When $K$ players are alive, $P_{n}$ shoots at some $P_{i}$ whose elimination results in
a $\left(  K-1\right)$-uel with highest payoff to $P_{n}$.
\end{quote}
What we have proved is that, for all $N\in \left \{  2,3,...\right \}  $ and all
$n\in \left \{  1,...,N\right \}$, the family $\left(\mathbf{R}_{K}\right)  _{K=1}^{N}$ yields a
deterministic NE for the $N$-uel.
\end{proof}

\bigskip

\noindent \noindent The above proof also furnishes an algorithm for computing
the equilibrium $\widehat{\sigma}=\left(  \widehat{\sigma}_{1},...,\widehat
{\sigma}_{N}\right)  $

\bigskip

\begin{minipage}[c]{0.95\linewidth}
\begin{algorithm}[H]
\caption{Computation of Nash Equilbrium}
\label{alg:cap}
\begin{algorithmic}
\Function{NashCompute}{$N,\mathbf{p}=\left(  p_{1},...,p_{N}\right)  $}
\State Initialize $\widehat{\sigma}_{1}\left(  \mathbf{s}\right)  $, $...,\widehat
{\sigma}_{N}\left(  \mathbf{s}\right)  $ for all $\mathbf{s}\in S_{1}\cup
S_{2}$
\For{$K\in \left \{  3,...,N\right \}  $}
\State Compute $\widehat{\sigma}_{1}\left(  \mathbf{s}\right)  $, $...,\widehat
{\sigma}_{N}\left(  \mathbf{s}\right)  $ for all $\mathbf{s}\in S_{K}$
\EndFor $K$.
\EndFunction
\end{algorithmic}
\end{algorithm}
\end{minipage}

\bigskip

\noindent If the original assumptions are violated there is no guarantee that
Algorithm 3 will yield the  equilibrium of the $N$-uel. However, it is
worth noting that if Algorithm 3 terminates, it will always yield an
equilibrium; i.e., the algorithm \emph{may work} for combinations of $p_{1}$,
..., $p_{N}$ values which violate some of our original assumptions (e.g., with
some marksmanships equal to zero or to one).

As will be seen in the next section, the strongest opponent rule can result in
rather interesting behaviors for certain combinations of $p_{1}$, ..., $p_{N}$.

\section{Experiments\label{sec05}}

In this section we use computer simulation to present some interesting
cases of $N$-uels. In all the following tables $\mathbf{s}_{n}$ denotes 
the state $\left(  n,1,1,1,\cdots ,1\right)  $.

\subsection{$3$-uels\label{sec0501}}

We start with $3$-uels in which, as mentioned, every player's optimal
strategy is to shoot at his strongest opponent (and this holds for all
states belonging to $S_{3}$).

\subsubsection{Strongest Player Has Highest / Lowest Winning Probability\label{sec050101}}

In Table \ref{tab01} we see a case where the strongest player $P_{\max}=P_1$, 
i.e., the one with highest markmanship, 
has the greatest expected payoff
regardless of who has the first move (note that, for each player, the optimal
strategy is the same for every state: shoot at your strongest opponent).
While this may seem natural, 
it is actually the exception and not the rule, as one might expect.

\begin{table}[ht]
\centering
\caption{Strongest player has the greatest expected payoff.}
{\begin{tabular}{@{}rrrr@{}} 
\hline
$n$ & $1$ & $2$ & $3$\\ \hline
$p_{n}$ & $0.90$ & $0.10$ & $0.20$\\ 
$\widehat{\sigma}_{n}\left(  \mathbf{s}_{n}\right)  $ & $3$ & $1$ & $1$\\ 
$V_{n,1111}$ & $0.86$ & $0.12$ & $0.02$\\ 
$V_{n,2111}$ & $0.62$ & $0.18$ & $0.20$\\ 
$V_{n,3111}$ & $0.69$ & $0.16$ & $0.15$\\ 
\hline
\end{tabular}
}
\label{tab01}
\end{table}

In Table \ref{tab02} we see that the player $P_{\max}=P_{3}$ does
\emph{not} have the highest expected payoff. In fact, if he does not have the
first move, he has the \emph{lowest} expected payoff. This is because each
player $P_{i}\neq P_{\max}$ will shoot at $P_{\max}$ (by the optimal strategy
of shooting at the strongest opponent) and he has a high probability of dying
before he has a chance to shoot back. Hence, the \textquotedblleft
team\textquotedblright \ consisting of the two players with lowest marksmanship
has a better survival probability than the $P_{\max}$ playing alone.

\begin{table}[ht]
\centering
\caption{Strongest player does not have  the greatest expected payoff.}
{\begin{tabular}{@{}rrrr@{}} 
\hline
$n$ & $1$ & $2$ & $3$\\ \hline
$p_{n}$ & $0.50$ & $0.70$ & $0.95$\\ \hline
$\widehat{\sigma}_{n}\left(  \mathbf{s}_{n}\right)  $ & $3$ & $3$ & $2$\\ 
$V_{n,1111}$ & $0.37$ & $0.56$ & $0.07$\\ 
$V_{n,2111}$ & $0.56$ & $0.30$ & $0.14$\\ 
$V_{n,3111}$ & $0.50$ & $0.03$ & $0.47$\\ 
\hline
\end{tabular}
}
\label{tab02}
\end{table}

\noindent This is one of many cases in which the strongest player $P_{3}$ has
the lowest expected payoff when he does not have the first move. And even when
he does have the first move, $P_{1}$ has greatest expected payoff. On the
other hand, $P_{1}$ has higher expected payoff when $P_{2}$ has the move, and
$P_{2}$ has higher  payoff when $P_{1}$ has the move.

\subsubsection{Zugzwang\label{sec050104}}

Consider the case $p_{1}=p_{2}=p_{3}=1$, i.e., every player has perfect
markmanship. Without loss of generality, we assume that $\hat{\sigma}%
_{n}=\mathbf{N}\left(  n111\right)  $ meaning that each $P_{n}$ shoots at the
next player (actually, it makes no difference to $P_{n}$ which player he will
shoot at). Let $P_{n}$ be the player who has the first move. $P_{n}$ will
always lose, no mater what strategy he uses, since, after killing his first
target he will fight a $2$-uel in which his opponent will have the first move
\emph{and} perfect markmanship and so will certainly kill $P_{n}$. These facts
are illustrated in Table \ref{tab03}. 

\begin{table}[H]
\centering
\caption{Player who has the first move loses.}
{\begin{tabular}{@{}rrrr@{}} 
\hline
$n$ & $1$ & $2$ & $3$\\ \hline
$p_{n}$ & $1.00$ & $1.00$ & $1.00$\\ 
$\widehat{\sigma}_{n}\left(  n111\right)  $ & $2$ & $3$ & $1$\\
$V_{n,1111}$ & $0.00$ & $0.00$ & $1.00$\\
$V_{n,2111}$ & $1.00$ & $0.00$ & $0.00$\\
$V_{n,3111}$ & $0.00$ & $1.00$ & $0.00$\\
\hline
\end{tabular}
}
\label{tab03}
\end{table}

This resembles a \emph{zugswang} position in
chess, i.e., a position in which a player will necessarily lose if he moves
\emph{any} of his pieces, whereas he would not necessarily lose if he could
pass (not move any piece).

\subsubsection{Being Weaker May Increase Payoff\label{sec050105}}

This example is a continuation of the previous one. In Figure \ref{fig04} we see that
$P_{1}$'s probability of winning is a \emph{decreasing} function of
marksmanship $p_{1}$ in the interval $\left[  0.5,1.0\right]  $, when
$p_{2}=p_{3}=1$.\textbf{ }In other words, having a lower marksmanship can
increase one's probability of winning (and this is true regardless of which
player has the first move).

\begin{figure}[hb]
\centerline{\includegraphics[width=3in]{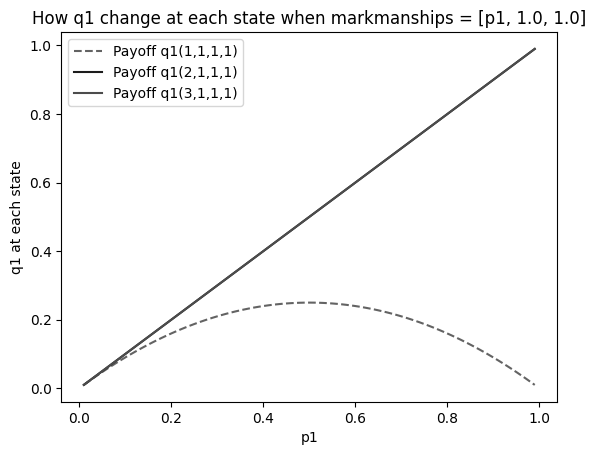}}
\caption{$P_1$'s payoff is a decreasing function ofhis marksmanship.}
\label{fig04}
\end{figure}

\subsection{4-uels\label{sec0502}}

\subsubsection{Shooting Weakest Opponent May Yield Maximum Payoff\label{sec050201}}

Adding one more player with $p_{4}<p_{1}$ at the experiment of Section
\ref{sec050105}, we end up with a paradoxical 4-uel. Consider the example of
the following Table \ref{tab04}.

\begin{table}[H]
\centering
\caption{Optimal strategy is to shoot at weakest player.}
{\begin{tabular}{@{}rrrrr@{}} 
\hline
$n$ & $1$ & $2$ & $3$ & $4$\\ \hline
$p_{n}$ & $0.70$ & $1.00$ & $1.00$ & $0.50$\\
$\widehat{\sigma}_{n}\left(  \mathbf{s}_{n}\right)  $ & $4$ & $3$ & $4$ & $2$\\ 
$V_{n,11111}$ & $0.66$ & $0.23$ & $0.00$ & $0.11$\\ 
\hline
\end{tabular}
}
\label{tab04}
\end{table}

\noindent$P_{1}$'s optimal strategy at $\mathbf{s=}11111$ is to shoot at the
weakest player. Taking the possible $3$-uels in which $P_{1}$ can end up if he
shoots successfully we have the following cases.

\begin{enumerate}
\item With $\sigma_{1}\left(  11111\right)  =2$, if $P_{1}$ kills $P_{2}$ he
ends up in a truel , similar to that of Section \ref{sec050101}, where $P_{3}$
with $p_{3}=1.00$ has the move and $P_{1}$ is the second best player. Hence,
the optimal strategy for $P_{3}$ is to shoot at $P_{1}\ $and $P_{1}$ loses in
the overall $4$-uel, i.e., $V_{3,1011}=0.$

\item With $\sigma_{1}\left(  11111\right)  =3$, $P_{1}$ ends up in a similar
$3$-uel where $P_{2}$ always shoots and kills $P_{1}$.

\item With $\sigma_{1}\left(  11111\right)  =4$, $P_{1}$ ends up in a truel
similar to the one of Section \ref{sec050105}, where he does not have the
first move and achieves $V_{1,21110}=0.66$. Hence this strategy yields the
best possible payoff to $P_{1}$.
\end{enumerate}

Note that in this example $P_{3}$'s best strategy is also to shoot at the
weakest player, because the probability of $P_{1}$ shooting successfully at
$P_{2}$ after that is high.

\subsubsection{Payoffs as Functions of Marksmanships $p_{1}$  and $p_{4}$ \label{sec050202}}

In this example we generalize the results of Section \ref{sec050201}. In
particular, we assume that $p_{2}=p_{3}=1$ and we study the dependence of the
payoffs to the \textquotedblleft nonperfect\textquotedblright \ players $P_{1}$
and $P_{4}$ on their marksmanships $p_{1}$ and $p_{4}$. 

\begin{figure}[ht]
\centerline{
  \subfigure[]
     {\includegraphics[width=2.5in]{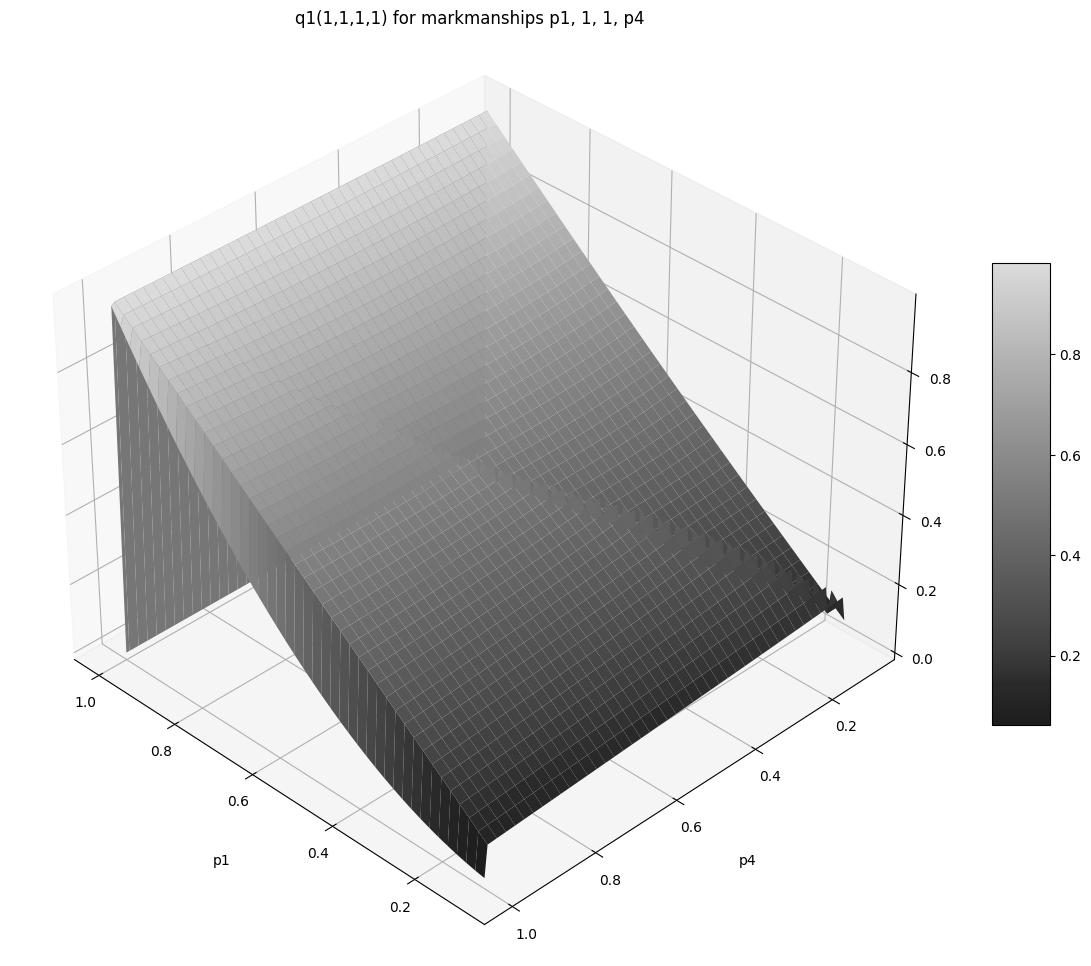}\label{fig05a}}
  \hspace*{4pt}
  \subfigure[]
     {\includegraphics[width=2.5in]{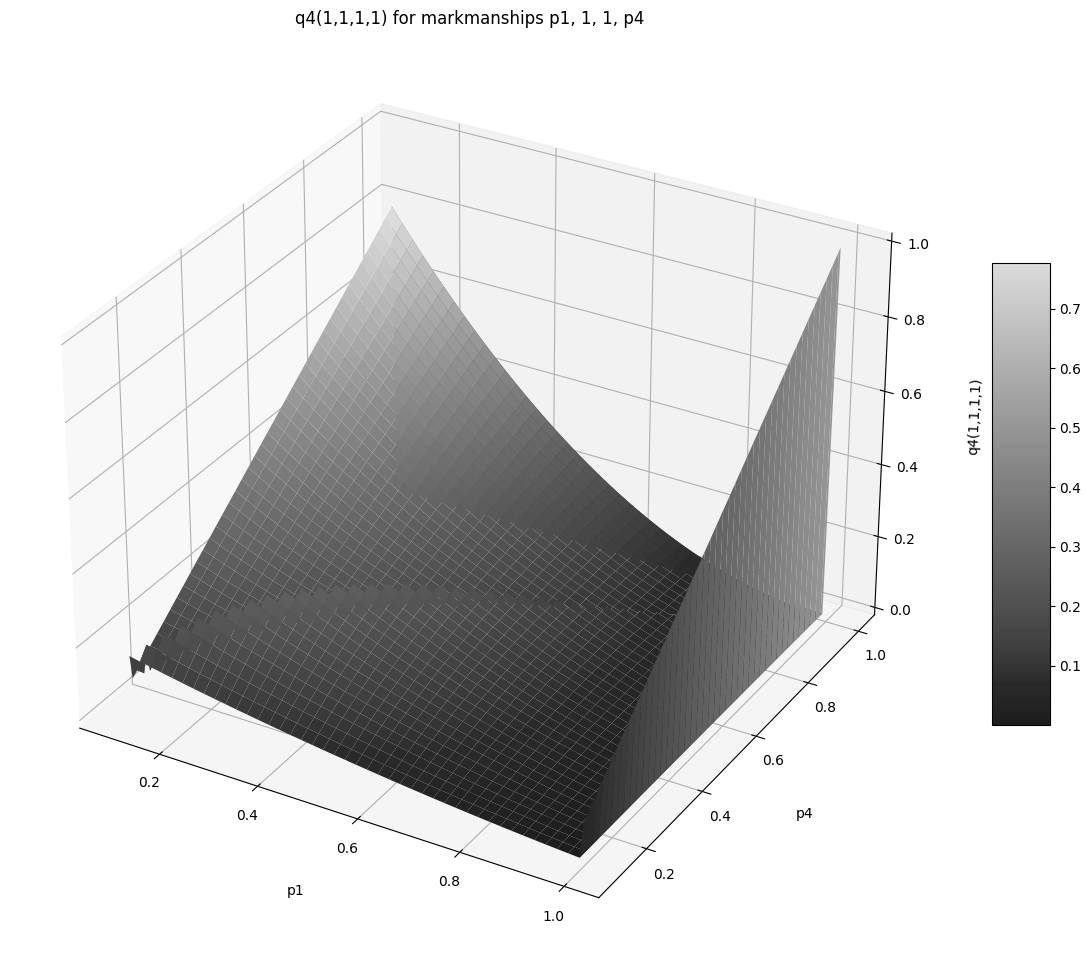}\label{fig05b}}
}
\caption{Payoffs $V_{1,11111}$ and $V_{4,11111}$ as functions of $p_{1}$ and $p_{4}$.}
\label{fig05}
\end{figure}

The surface in Figure \ref{fig05a} (resp. in Figure \ref{fig05b}) 
is $V_{1,11111}$ (resp. $V_{4,11111}$) as a function of $p_{1}$ and $p_{4}$, when $p_{2}=p_{3}=1$. 
In both figures
we have a discontinuity at $p_{1}=p_{4}$. This is due to the fact that players
change strategies. Taking $p_{4}>p_{1}$ in the $3$-uel starting at, for
example, $\mathbf{s}=21101$, $P_{4}$ is $P_{2}$'s new target as he is the next
strongest player after $P_{2}$.

\subsubsection{Circular Shooting Sequences and Formation of Teams\label{sec050203}}

We now present two examples in which we focus on the players' strategies
rather than their payoffs. Taking $\mathbf{p=}\left(
0.80,0.40,0.85,0.50\right)  $, we get the following optimal strategies.

\begin{table}[H]
\centering
\caption{An example of circular shooting and team formation.}
{\begin{tabular}{@{}rrrrr@{}} 
\hline
$n$ & $1$ & $2$ & $3$ & $4$\\ \hline
$p_{n}$ & $0.80$ & $0.40$ & $0.85$ & $0.50$\\ 
$\widehat{\sigma}_{n}\left(  \mathbf{s}_{n}\right)  $ & $4$ & $3$ & $1$ & $2$\\
\hline
\end{tabular}
}
\label{tab05}
\end{table}

\noindent Taking $\mathbf{p=}\left(  0.75,0.25,1.00,0.50\right)  $, we get the
following optimal strategies.

\begin{table}[H]
\centering
\caption{Another example of circular shooting and team formation.}
{\begin{tabular}{@{}rrrrr@{}} 
\hline
$n$ & $1$ & $2$ & $3$ & $4$\\ \hline
$p_{n}$ & $0.75$ & $0.25$ & $1.00$ & $0.50$\\ 
$\widehat{\sigma}_{n}\left(  \mathbf{s}_{n}\right)  $ & $4$ & $3$ & $1$ & $2$\\ 
\hline
\end{tabular}
}
\label{tab06}
\end{table}

\noindent In both cases the optimal strategy profile follows a \emph{circular
shooting order:}%
\[
P_{1}\underset{\text{shoots at}}{\rightarrow}P_{4}\underset{\text{shoots at}%
}{\rightarrow}P_{2}\underset{\text{shoots at}}{\rightarrow}P_{3}%
\underset{\text{shoots at}}{\rightarrow}P_{1}%
\]
More generally, in all $4$-uel experiments where we have a circular strategy,
we can consider the players to be forming two teams. In this experiment the
teams are $\{P_{1},P_{2}\}$ and $\{P_{3},P_{4}\}$; the optimal strategy for
each member of each team is to shoot at his teammate's shooter.

\begin{center}

\end{center}

\subsubsection{Solidarity of the Weakest}

In our final $4$-uel example, taking $\mathbf{p=}\left(
0.05,0.10,0.15,0.70\right) $ we get a situation which resembles an economic
model with three small businesses $P_{1},P_{2}$ and $P_{3}$ and a large one
$P_{4}$. The optimal strategy for the large business $P_{4}$ is to eliminate
his strongest opponent, hence he shoots at $P_{3}$. The other three players
are so weak that they do not want to lead the game to a truel similar to the
experiment of Section \ref{sec050101}, where the strongest player has the
greatest payoff (for example, $P_{2}$ will not shoot $P_{1}$). Consequently,
the three weak players ($P_{1},P_{2}$ and $P_{3}$)\ cooperate to eliminate the
strong player. The optimal strategies and respective payoffs are shown in
Table \ref{tab07}. Note that $P_{4}$ has the greater expected payoff in each case and,
in fact, his payoff is either very close or higher than $0.5$.

\begin{table}[H]
\centering
\caption{Solidarity of the weakest players.}
{\begin{tabular}{@{}rrrrr@{}} 
\hline
$n$ & $1$ & $2$ & $3$ & $4$\\ \hline
$p_{n}$ & $0.05$ & $0.10$ & $0.15$ & $0.70$\\ 
$\widehat{\sigma}_{n}\left(  \mathbf{s}_{n}\right)  $ & $4$ & $4$ & $4$ & $3$\\ 
$V_{n,11111}$ & $0.18$ & $0.20$ & $0.13$ & $0.49$\\ 
$V_{n,21111}$ & $0.18$ & $0.19$ & $0.12$ & $0.51$\\ 
$V_{n,31111}$ & $0.16$ & $0.18$ & $0.09$ & $0.57$\\ 
$V_{n,41111}$ & $0.14$ & $0.15$ & $0.04$ & $0.67$\\ 
\hline
\end{tabular}
}
\label{tab07}
\end{table}

In $N$-uels of this type, we see a \emph{solidarity} effect between the
weakest players. However, it must be noted that such $N$-uels are a small, not
representative, subset of the possible cases, as we have seen from our
previous examples.

\subsection{N-uels, $N\geq5$}

We conclude with two examples where, similar to the example of Section
\ref{sec050203}, we obtain circular shooting sequences. In the $5$-uel with
$\mathbf{p}=\left(  0.25,0.20,0.10,0.06,0.04\right)  $, computation shows that
each player's optimal strategy is to shoot at the next player and the last
shoot at the first.

\begin{table}[H]
\centering
\caption{An example of circular shooting.}
{\begin{tabular}{@{}rrrrrr@{}} 
\hline
$n$ & $1$ & $2$ & $3$ & $4$ & $5$\\ \hline
$p_{n}$ & $0.25$ & $0.20$ & $0.10$ & $0.06$ & $0.04$\\ 
$\widehat{\sigma}_{n}\left(  \mathbf{s}_{n}\right)  $ & $2$ & $3$ & $4$ & $5$ & $1$\\
\hline
\end{tabular}
}
\label{tab08}
\end{table}

\noindent Here is a $7$-uel which has a similar effect.

\begin{table}[H]
\centering
\caption{An example of circular shooting.}
{\begin{tabular}{@{}rrrrrrrr@{}} 
\hline
$n$ & $1$ & $2$ & $3$ & $4$ & $5$ & $6$ & $7$\\ \hline
$p_{n}$ & $0.40$ & $0.26$ & $0.18$ & $0.12$ & $0.08$ & $0.06$ & $0.04$\\ 
$\widehat{\sigma}_{n}\left(  \mathbf{s}_{n}\right)  $ & $4$ & $1$ & $5$ & $6$ & $2$ & $7$ & $3$\\
\hline
\end{tabular}
}
\label{tab09}
\end{table}

\noindent While we have discovered many similar examples by brute force
computation, we have not been able to obtain a condition on the $\mathbf{p}$
values which guarantees the emergence of circular shooting orders. Also, it is
not obvious what shooting order will emerge as soon as one player is
eliminated; it is not usually the case that the resultant shooting order will
again be circular. In the future we intend to further study these questions.

\section{Conclusion\label{sec06}}

We have studied a $N$-uel game (a generalization of the \emph{duel}) in which
finding a Nash equilibrium reduces to solving the system of nonlinear payoff
equations. We have proved that this system has a  solution (hence the
$N$-uel has a  stationary Nash equilibrium) and we have provided
algorithms for its computational solution. In the future we want to study

\begin{enumerate}
\item the existence and computation of \emph{nonstationary} Nash equilibria.

\item the properties of the $N$-uel variant in which each player can
\emph{abstain}, i.e., not shoot at any of his opponents.
\end{enumerate}

\bibliographystyle{ws-rv-van}

\begin{thebibliography}{99}                                                                                               %


\bibitem {Amengual2005a}Amengual, P. and Toral, R. \textquotedblleft
Distribution of winners in truel games\textquotedblright. \emph{AIP Conference
Proceedings} (2005), pp. 128--141.

\bibitem {Amengual2005b}Amengual, P. and Toral, R. \textquotedblleft A Markov
chain analysis of truels\textquotedblright. \emph{Proceedings of the 8th
Granada Seminar on Computational Physics} (2005).

\bibitem {Amengual2006}Amengual, P. and Toral, R. \textquotedblleft Truels, or
survival of the weakest.\textquotedblright \  \emph{Computing in Science}.

\bibitem {Barron2013}Barron, E. N. \emph{Game theory: an introduction}. John
Wiley \& Sons (2013).

\bibitem {Bossert2002}Bossert, W., Brams S.J. and Kilgour, D.M.
\textquotedblleft Cooperative vs non-cooperative truels: little agreement, but
does that matter?.\textquotedblright \  \emph{Games and Economic Behavior}, vol.
40 (2002), pp. 185-202.

\bibitem {Brams1997}Brams, S.J., and Kilgour, D.M. \textquotedblleft The
truel\textquotedblright. \emph{Mathematics Magazine}, vol. 70 (1997), pp. 315-326.

\bibitem {Brams2001}Brams, S.J., and Kilgour, D.M. \textquotedblleft Games
that End in a Bang or a Whimper\textquotedblright. preprint, CV Starr Center
for Applied Economics (2001).

\bibitem {Brams2003}Brams, S.J., Kilgour, D.M. \ and Dawson, B.
\textquotedblleft Truels and the Future.\textquotedblright \  \emph{Math
Horizons}, vol. 10 (2003), pp. 5-8.

\bibitem {Dorraki2019}Dorraki, M., Allison, A. and Abbott, D..
\textquotedblleft Truels and strategies for survival.\textquotedblright%
\  \emph{Scientific Reports}, vol. 9 (2019), pp. 1-7.

\bibitem {Dresher1961}Dresher, M. \emph{Games of strategy: theory and
applications}. Rand Corp. (1961).

\bibitem {Gardner1966}Gardner, M. : \emph{New Mathematical Puzzles and
Diversions}. (1966), pp. 42-49.

\bibitem {Karlin1959}Karlin, S., \emph{Mathematical Methods and Theory in
Games, Programming, and Economics}, vol. 2, Addison- Wesley (1959).

\bibitem {Kilgour1971}Kilgour, D. M. \textquotedblleft The simultaneous
truel.\textquotedblright \  \emph{International Journal of Game Theory}, vol. 1
(1971), pp. 229-242.

\bibitem {Kilgour1975}Kilgour, D. M. \textquotedblleft The sequential
truel.\textquotedblright \  \emph{International Journal of Game Theory}, vol. 4
(1975), pp. 151-174.

\bibitem {Kilgour1977}Kilgour, D. M. \textquotedblleft Equilibrium points of
infinite sequential truels.\textquotedblright \  \emph{International Journal of
Game Theory}, vol. 6 (1977), pp. 167-180.

\bibitem {Kinnaird1946}Kinnaird, C. \emph{Encyclopedia of puzzles and
pastimes}. Secaucus (1946).

\bibitem {Knuth1973}Knuth, D.E. \textquotedblleft The Triel: A New
Solution.\textquotedblright \  \emph{Journal of Recreational Mathematics}, vol.
6 (1972), pp. 1-7.

\bibitem {Larsen1948}Larsen, H.D. \textquotedblleft A Dart
Game.\textquotedblright \  \emph{American Mathematical Monthly}, (1948) pp. 640-41.

\bibitem {Mosteller1987}Mosteller, F. \emph{Fifty challenging problems in
probability with solutions}. Courier Corporation (1987).

\bibitem {Shubik1954}Shubik, M. \textquotedblleft Does the Fittest Necessarily
Survive?\textquotedblright, \emph{Readings in Game Theory and Political
Behavior}, Doubleday.

\bibitem {Sobel1971}Sobel, M.J. \textquotedblleft Noncooperative stochastic
games.\textquotedblright \  \emph{The Annals of Mathematical Statistics}, vol.
42 (1971), pp. 1930-1935.

\bibitem {Toral2005}Toral, R., and Amengual, P.. \textquotedblleft
Distribution of winners in truel games.\textquotedblright \  \emph{AIP
Conference Proceedings}, vol. 779. American Institute of Physics (2005).

\bibitem {Wegener2021}Wegener, M. and Mutlu, E.. \textquotedblleft The good,
the bad, the well-connected.\textquotedblright \  \emph{International Journal of
Game Theory}, vol. 50 (2021), pp. 759-771.

\bibitem {Xu2012}Xu, X. \textquotedblleft Game of the truel.\textquotedblright%
\  \emph{Synthese}, vol. 185 (2012), pp. 19-25.

\bibitem {Zeephongsekul1991}Zeephongsekul, P. \textquotedblleft Nash
Equilibrium Points Of Stochastic $N$-Uels.\textquotedblright \  \emph{Recent
Developments in Mathematical Programming}. CRC Press (1991), pp. 425-452.
\end{thebibliography}

\end{document}